\documentclass[a4paper,11pt]{article}

\usepackage[Postscipt=dvips]{diagrams}
\usepackage{latexsym}
\usepackage{amssymb}
\usepackage{stmaryrd}
\usepackage{proof}

\title{A Game Semantics for Generic Polymorphism}

\author{Samson Abramsky\\
Oxford University Computing Laboratory\\
\and
Radha Jagadeesan\\
DePaul University}
\date{}

\newcommand{\gequiv}{\approx}
\newcommand{\lsem}{\llbracket}
\newcommand{\rsem}{\rrbracket}
\newcommand{\eqdef}{\stackrel{\vartriangle}{=}}
\newcommand{\UU}{\mathcal{U}}
\newtheorem{proposition}{Proposition}[section]
\newtheorem{lemma}[proposition]{Lemma}
\newtheorem{theorem}[proposition]{Theorem}
\newtheorem{corollary}[proposition]{Corollary}
\newcommand{\GG}[1]{\mathcal{G}(#1)}
\newcommand{\GasS}{\mathbf{Games}}
\newcommand{\numoccs}[2]{\mathsf{numoccs}(#1 , #2 )}
\newcommand{\parity}{\mathsf{parity}}
\newcommand{\restrict}{{\upharpoonright}}
\newcommand{\llwith}{\, \& \,}
\newcommand{\domapprox}{\, \sqsubseteq \,}
\newcommand{\ginc}{\, \trianglelefteq \,}
\newcommand{\Occ}{\mathcal{O}}
\newcommand{\PInv}[1]{\mathsf{PInv}(#1)}
\newcommand{\twist}{\mathsf{twist}}
\newcommand{\Rel}[1]{\mathsf{Rel}(#1)}
\newcommand{\pow}[1]{\mathcal{P}(#1)}
\newcommand{\invcomp}{\bowtie}
\newcommand{\rconv}[1]{#1^{\mathsf{c}}}
\newcommand{\rdconv}[1]{#1^{\mathsf{c}\, \mathsf{c}}}
\newcommand{\preord}{\lessapprox}
\newcommand{\ident}{\mathsf{id}}
\newcommand{\ICC}{\mathcal{C}}
\newcommand{\CCC}{\mathbf{CCC}}
\newcommand{\CCB}{\ICC_{B}}
\newcommand{\GB}[1]{\mathcal{G}_{B}(#1)}
\newcommand{\GU}[1]{\mathcal{G}_{\UU}(#1)}
\newcommand{\SubU}{\mathsf{Sub}(\UU )}
\newcommand{\SubB}{\mathsf{Sub}(B)}
\newcommand{\Base}{\mathbb{B}}
\newcommand{\List}[1]{\mathsf{List}[#1]}
\newcommand{\Pref}{\mathsf{Pref}}
\newcommand{\ie}{\textit{i.e.}\ }
\newcommand{\pfnarrow}{\; - \! \! \! \rightharpoonup \;}
\newcommand{\al}{\mathtt{p}}
\newcommand{\ar}{\mathtt{q}}
\newcommand{\fr}{\mathtt{r}}
\newcommand{\fl}[1]{\mathtt{l}_{#1}}
\newcommand{\hext}[1]{#1^{\dagger}}
\newcommand{\MM}{\mathcal{M}}
\newenvironment{proof}{\textsc{Proof}\ }{$\;\; \Box$}

\newcommand{\bangindex}[1]{\mathtt{k}_{#1}}
\newcommand{\linearfl}{\mathtt{l}}
\newcommand{\tensorfl}{\mathtt{l^t}}
\newcommand{\tensorfr}{\mathtt{r^t}}
\newcommand{\tensor}{\otimes}
\newcommand{\linimpl}{\multimap}
\newcommand{\LUU}{\mathcal{U'}}
\newcommand{\SubLU}{\mathsf{Sub}(\LUU )}
\newcommand{\linearMM}{\mathcal{M'}}
\newcommand{\lambdaomega}{\lambda\Omega}

\newcommand{\lang}{\langle}
\newcommand{\rang}{\rangle}

\begin{document}

\maketitle

\begin{center}
\textit{Dedicated to Helmut Schwichtenberg on the occasion of his
  sixtieth birthday. His commitment to the highest scientific
  standards, coupled with a wise and kind humanity, is a continuing
  source of inspiration.}
\end{center}

\begin{abstract}
Genericity is the idea that the same program can work at many
different data types.  Longo, Milstead and Soloviev proposed to capture the
inability of generic programs to probe the structure of their
instances by the following equational principle: if two generic
programs, viewed as terms of type $\forall X. \, A[X]$, are equal
at any given instance $A[T]$, then they are equal at all
instances. They proved that this rule is admissible in a certain
extension of System F, but finding a semantically motivated model
satisfying this principle remained an open problem.

In the present paper, we construct a categorical
model of polymorphism, based on game semantics, which contains a
large collection of generic types. This model builds on two novel
constructions:
\begin{itemize}
\item A direct interpretation of variable types as games, with a
  natural notion of substitution of games. This allows moves in
  games $A[T]$ to be decomposed into the generic part from $A$, and
  the part pertaining to the instance $T$. This leads to a
  simple and natural notion of generic strategy.

\item A ``relative polymorphic product''
  $\Pi_i (A, B)$ which expresses quantification over the type variable
  $X_i$ in the variable type $A$ with respect to a ``universe'' which
  is explicitly given as an additional parameter $B$. We then solve a
  recursive equation involving this relative product to obtain a
  universe in a suitably ``absolute'' sense.
  \end{itemize}
Full Completeness for ML types (universal closures of
quantifier-free types) is proved for this model.
\end{abstract}

\section{Introduction}
We begin with an illuminating quotation from G\'erard Berry \cite{Ber00}:
\begin{quotation}
Although it is not always made explicit, the \textit{Write Things
  Once} or WTO principle is clearly the basis for loops, procedures,
higher-order functions, object-oriented programming and inheritance,
concurrency \textit{vs.} choice between interleavings, etc.
\end{quotation}
In short, much of the search for high-level structure in programming
can be seen as the search for concepts which allow commonality to be expressed. An important facet of this
quest concerns \emph{genericity}: the idea that the same program can
work at many different data types.

For illustration, consider the abstraction step involved in passing
from list-processing programs which work on data types $\List{T}$ for
specific types $T$, to programs which work generically on $\List{X}$.
Since lists can be so clearly visualized,
\begin{figure}
\begin{center}

\begin{picture}(400,60)(-50,-30)
\put(-40,10){\vector(1,0){40}}                  
\multiput(0,0)(100,0){2}{\line(1,0){40}}       
\multiput(0,20)(100,0){2}{\line(1,0){40}}       
\multiput(0,0)(100,0){2}{\line(0,1){20}}        
\multiput(20,0)(100,0){2}{\line(0,1){20}}       
\multiput(40,0)(100,0){2}{\line(0,1){20}}       
\multiput(10,10)(100,0){2}{\vector(0,-1){40}}  
\multiput(30,10)(100,0){2}{\vector(1,0){70}}    
\multiput(205,10)(10.00000,0.00000){8}{\line( 1, 0){5.000}}
\put(260,10){\vector(1,0){40}}                  
\put(320,0){\line(1,1){20}}                    
\put(320,20){\line(1,-1){20}}                  
\put(300,0){\line(1,0){40}}                   
\put(300,20){\line(1,0){40}}                 
\put(300,0){\line(0,1){20}}        
\put(320,0){\line(0,1){20}}       
\put(340,0){\line(0,1){20}}       
\put(310,10){\vector(0,-1){40}}     
\end{picture}

\end{center}
\caption{`Generic' list structure}
\end{figure}
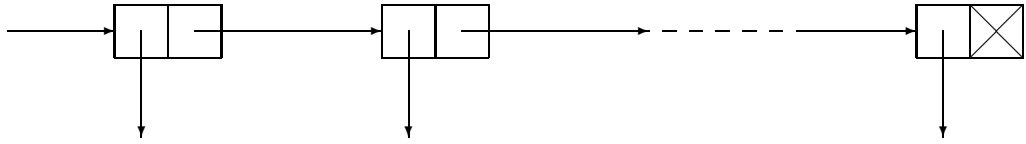
it is easy to see what this should mean (see Figure 1). A generic program cannot
probe the internal structure of the list elements. Thus e.g. list
concatenation and reversal are generic, while summing a list is
not. However, when we go beyond lists and other concrete data structures,
to higher-order types and beyond, what genericity or
type-independence should mean becomes much less clear.

One very influential proposal for a general understanding of the
\emph{uniformity} which generic programs should exhibit with respect
to the  type instances has been John Reynolds' notion of
\emph{relational parametricity} \cite{Rey83}, which requires that
relations between  instances be preserved in a suitable sense
by generic programs. This has led to numerous further
developments, e.g. \cite{MR92,ACC93,PA93}.

Relational parametricity is a beautiful and important notion. However,
in our view it is not the whole story. In particular:
\begin{itemize}
\item It is a ``pointwise'' notion, which gets at genericity
  indirectly, via a notion of uniformity applied to the family of
  instantiations of the program, rather than directly
  capturing the idea of a  program written at the generic level,
  which necessarily cannot probe the structure
  of an instance.
\item It is closely linked to strong extensionality principles, as
  shown e.g. in \cite{ACC93,PA93}, whereas the intuition of generic programs not
  probing the structure of instances is \textit{prima facie} an
intensional notion---a constraint on the behaviour of processes.
\end{itemize}
An interestingly different analysis of genericity with different
formal consequences was proposed by Giuseppe Longo, Kathleen Milsted
and Sergei Soloviev \cite{LMS93,Lon95}. Their idea
was to capture the inability of generic programs to probe the
structure of their instances by the following equational principle: if
two generic programs, viewed as terms $t$, $u$ of type $A[X]$, are
equal at \emph{any} given instance $T$, then they are equal at
\emph{all} instances:
\[ \exists T. \, t\{ T \} = u\{ T \}  :  A[T] \;\;
\Longrightarrow \;\; \forall U. \, t \{ U \} = u \{ U \} : A[U] . \]
This principle can be stated even more strongly when second-order
polymorphic quantification over type variables is used. For $t, u :
\forall X. \, A$:
\[  \frac{t\{ T \} = u\{ T \}  :  A[T]}
{t  = u  : \forall X. \, A}. \]
We call this the \emph{Genericity Rule}.
In one of the most striking syntactic results obtained for System F
(\ie the polymorphic second-order $\lambda$-calculus \cite{Gir72,Rey74}), Longo, Milsted
and Soloviev proved in \cite{LMS93} that the Genericity Rule is admissible in the system
obtained by extending System F with the following axiom scheme:
\[ (C) \qquad t \{ B \} = t \{ C \} : A \qquad (t : \forall X. \, A, \; X \not\in
\mathrm{FV}(A)) . \]
While many of the known semantic models of System F satisfy axiom (C),
\emph{there is no known naturally occurring model which satisfies the
  Genericity principle} (\ie in which the  rule of
  Genericity is valid). In fact, in the strong form given above, the
  Genericity rule is actually \emph{incompatible} with
  well-pointedness and parametricity, as observed by Longo. Thus if we
  take the standard polymorphic terms representing the Boolean values
\[ \Lambda X.\, \lambda x{:}X.\, \lambda y{:}X.\, x , \;\; \Lambda X. \, \lambda
x{:}X. \, \lambda y{:}X.\, y \;\; : \;\; \forall X. \, X \rightarrow X \rightarrow X \]
then if the type $\forall X. \, X \rightarrow X$ has only one
inhabitant --- as will be the case in a parametric model --- then by
well-pointedness the Boolean values will be equated at this instance,
while they cannot be equated in general on pain of inconsistency.

However, we can state a more refined version. Say that a type $T$ is
\emph{a generic instance} if for all types $A[X]$:
\[ t\{ T\} = u\{T\} : A[T] \;\; \Longrightarrow \;\; t = u : \forall X. \, A
. \]
This leads to the following problem posed by Longo in  \cite{Lon95}, and
still, to the best of our knowledge, open:
\begin{quotation}
Open Problem 2. Construct, at least, some (categorical) models that
contain a collection of ``generic'' types. \ldots If our intuition
about constructivity is correct, infinite objects in categories of
(effective) sets should satisfy this property.
\end{quotation}

In the present paper, we present a solution to this problem by
constructing a categorical model of polymorphism which contains a
large collection of generic types. The model is based on game
semantics; more precisely, it extends the ``AJM games'' of \cite{AJM00} to
provide a model for generic polymorphism. Moreover, Longo's intuition
as expressed above is confirmed in the following sense: our main
sufficient condition for games (as denotations of types) to be generic
instances is that they have plays of arbitrary length. This can be
seen as an intensional version of Longo's intuition about infinite
objects.

In addition to providing a solution to this problem, the present paper
also makes the following contributions.
\begin{itemize}
\item We interpret variable types in a simple and direct way, with a
  natural notion of \emph{substitution of games into variable
    games}. The crucial aspect of this idea is that it allows moves in
  games $A[T]$ to be decomposed into the generic part from $A$, and
  the part pertaining to the instance $T$. This in turn allows the
  evident content of genericity in the case of concrete data
  structures such as lists to be carried over to arbitrary
  higher-order and polymorphic types. In particular, we obtain a
  simple and natural notion of \emph{generic strategy}. This extends
  the notion of history-free strategy from \cite{AJM00}, which is
  determined by a function on moves, to that of a generic strategy,
  which is determined by a function on \emph{the generic part of the move only},
and simply acts as the identity on the part pertaining to the
instance. This captures the intuitive idea of a generic program, existing ``in
advance'' of its instances, in a rather direct way.

\item We solve the size problem inherent in modelling System F in a
  somewhat novel way. We define a ``relative polymorphic product''
  $\Pi_i (A, B)$ which expresses quantification over the type variable
  $X_i$ in the variable type $A$ with respect to a ``universe'' which
  is explicitly given as an additional parameter $B$. We then solve a
  recursive equation involving this relative product to obtain a
  universe in a suitably ``absolute'' sense: a game $\UU$
  with the requisite closure properties to provide a model for System
  F.

\item We prove Full Completeness for the ML types (\ie the universal
  closures of quantifier-free types).

\end{itemize}

\section{Background}

\subsection{Syntax of System F}
We briefly review the syntax of System F. For further background
information we refer to \cite{GLT89}.

\subsubsection*{Types (Formulas)}
\[ A \quad ::= \quad X \; \mid \; A \rightarrow B \; \mid \; \forall X. \, A \]

\subsubsection*{Typing Judgements}
Terms in context have the form
\[ x_1 : A_1 , \ldots , x_k : A_k \vdash t : A \]
\textbf{Assumption}
\[ \infer{\Gamma , x:T \vdash  x:T}{} \]

\noindent \textbf{Implication}
\[ \infer[({\rightarrow}-I)]{\Gamma \vdash \lambda x{:}U. \, t: U \rightarrow T}{\Gamma , x:U \vdash t:T}
\qquad
\infer[({\rightarrow}-E)]{\Gamma \vdash tu : T}{\Gamma \vdash t : U \rightarrow T \qquad \Gamma
  \vdash u:U}
\]

\noindent \textbf{Second-order Quantification}
\[  \infer[({\forall}-I)]{\Gamma \vdash \Lambda X. \, t : \forall
  X. \, A}{\Gamma \vdash t:A}  \qquad
\infer[({\forall}-E)]{\Gamma \vdash t\{ B \} :
  A[B/X]}{\Gamma \vdash t : \forall X. \, A}
\]
The $(\forall-I)$ rule is subject to the usual eigenvariable
condition, that $X$ does not occur free in $\Gamma$.

\noindent The following isomorphism is definable in System F:
\[ \forall X. \, A \rightarrow B \;\; \cong \;\; A \rightarrow \forall
X. \, B \qquad (X \not\in \mathrm{FV}(A)) . \]
This allows us to use the following normal form for types:
\[ \forall \vec{X}. \, T_1 \rightarrow \cdots \rightarrow T_k
\rightarrow X \qquad (k \geq 0) \]
where each $T_i$ is inductively of the same form.

\subsection{Notation}
We write $\omega$ for the set of natural numbers.

If $X$ is a set, $X^*$ is the set of finite sequences (words, strings)
over $X$.  We use $s$, $t$, $u$, $v$ to denote sequences, and $a$,
$b$, $c$, $d$, $m$, $n$ to denote elements of these sequences.
Concatenation of sequences is indicated by juxtaposition, and we don't
distinguish notationally between an element and the corresponding unit
sequence. Thus $as$ denotes the sequence with first element $a$ and
tail
$s$. However, we will sometimes write $a \cdot s$ or $s \cdot a$ to
give the name $a$ to the first or last element of a sequence.

If $f:X\longrightarrow Y$ then $f^*:X^*\longrightarrow Y^*$ is the
unique
monoid homomorphism extending $f$.  We write $|s|$ for the length of a
finite sequence, and $s_i$ for the $i$th element of $s$,
$1\leq i\leq |s|$. We write $\numoccs{a}{s}$ for the number of occurrences of $a$ in the sequence $s$.

We write $X+Y$ for the disjoint union of sets $X$, $Y$.

If $Y\subseteq X$ and $s\in X^*$, we write $s\restriction Y$ for the
sequence obtained by deleting all elements not in $Y$ from $s$. In
practice, we use this notation in the context where $X = Y + Z$, and
by abuse of notation we take $s \restriction Y \in Y^{*}$, {\ie}\ we
elide the use of injection functions.
We also use several variations on the notion of projection onto a sub-sequence, defining any which are not obvious from the context.

We
write $s\sqsubseteq t$ if $s$ is a prefix of $t$, {\em i.e.}\ $t=su$
for some $u$. We write $s\sqsubseteq^{\mbox{{\scriptsize even}}} t$ if $s$ is an even-length prefix of $t$.
$\Pref(S)$ is the set of prefixes of elements of $S\subseteq X^*$.
$S$ is {\em prefix-closed} if $S=\Pref(S)$.

\section{Variable Games and Substitution}
\subsection{A Universe of Moves}
We fix an algebraic signature consisting of the following set of
\emph{unary} operations:
\[ \al , \; \ar , \; \{ \fl{i} \mid i \in \omega \} , \; \fr . \]
We take $\MM$ to be the algebra over this signature freely generated
by $\omega$. Explicitly, $\MM$ has the following ``concrete syntax'':
\[ m \quad ::= \quad i \; (i \in \omega ) \;\; \mid \;\; \al (m) \;\;
\mid \;\; \ar (m) \;\; \mid \;\;
\fl{i} (m) \;\; (i \in \omega ) \;\; \mid \;\; \fr (m) . \]
For any algebra $(A, \al^A , \ar^A , \{ \fl{i}^A \mid i \in \omega \}
, \fr^A )$ and map $f : \omega \longrightarrow A$, there is a unique
homomorphism $ \hext{f} : \MM \longrightarrow A$ extending $f$, defined
by:
\[ \hext{f}(i) = f(i), \qquad \hext{f}(\phi (m)) = \phi^A (\hext{f}(m))
\quad (\phi \in \{ \al , \ar , \fr \} \cup \{ \fl{i} \mid i \in \omega
\} ). \]
We now define a number of maps on $\MM$ by this means.
\begin{itemize}
\item
The \emph{labelling map} $\lambda : \MM \longrightarrow \{ P, O \}$.
The polarity algebra on the carrier $\{ P, O \}$ interprets $\al$,
$\ar$, $\fr$ as the identity, and each $\fl{i}$ as the involution
$\bar{(\ )}$, where $\bar{P} = O$, $\bar{O} = P$. The map on the
generators is the constant map sending each $i$ to $O$.

\item
The map $\rho : \MM \longrightarrow \omega$ sends each move to the
unique generator occurring in it. All the unary operations are
interpreted as the identity, and the map on generators is the
identity.

\item
The substitution map. For each move $m' \in \MM$, there is a map
\[ h_{m'} : \MM \longrightarrow \MM \]
induced by the constant map on $\omega$ which sends each $i$ to
$m'$. We write $m[m']$ for $h_{m'}(m)$.

\item
An alternative form of substitution is written $m[m'/i]$. This is
induced by the map which sends $i$ to $m'$, and is the identity on
all $j \neq i$.
\end{itemize}

\begin{proposition}
Substitution is associative and left-cancellative:
\[ \begin{array}{ll}
(1) & m_1[m_2[m_3]] = (m_1[m_2])[m_3] \\
(2) & m[m_1] = m[m_2] \;\; \Longrightarrow \;\; m_1 = m_2
\end{array} \]
\end{proposition}
Note that substitution is right-cancellative \emph{only up to
  permutation of generators}:
\[ m[i][m'] = m[m'] = m[j][m'] \qquad \mbox{for all $i, j \in
  \omega$}. \]

\begin{proposition}
\label{sublr}
Substitution interacts with $\lambda$ and $\rho$ as follows.
\[ \begin{array}{ll}
1. & \lambda (m[m']) = \left\{ \begin{array}{ll}
\overline{\lambda (m')} & \mbox{if $\lambda (m) = P$} \\
\lambda (m') & \mbox{if $\lambda (m) = O$}
\end{array} \right. \\
2. & \rho (m[m']) = \rho (m') .
\end{array} \]
\end{proposition}

\noindent We extend the notions of substitution pointwise to sequences and sets
of sequences of moves in the evident fashion.

\noindent We say that $m_1 , m_2 \in \MM$ are \emph{unifiable} if for
some $m_3 , m_4
\in \MM$, $m_1 [m_3] = m_2 [m_4]$. A set $S \subseteq \MM$ is
\emph{unambiguous} if whenever $m_1 , m_2 \in S$ are unifiable, $m_1 =
m_2$.

\begin{proposition}
\label{unamb}
If $S$ is unambiguous,
and for each $m \in S$ the set $T_m$ is unambiguous, then so is the
following set:
\[ \{ m_1[m_2] \mid m_1 \in S \; \wedge \; m_2 \in T_{m_1} \} . \]
\end{proposition}
\begin{proof}
Suppose that $(m_1[m_2])[m_3] = (m'_1[m'_2])[m'_3]$. We must show
that $m_1[m_2] = m'_1[m'_2]$. By associativity, $m_1[m_2[m_3]] =
m'_1[m'_2[m'_3]]$. Since $S$ is unambiguous, this implies that
$m_1 = m'_1$. By left cancellativity, this implies that $m_2[m_3]
= m'_2[m'_3]$. Since $T_{m_1}$ is unambiguous, this implies that
$m_2 = m'_2$.
\end{proof}

\noindent Given a subset $S \subseteq \MM$ and $i \in \omega$, we write
\[ S^i = \{ m \in S \mid \rho (m) = i \} . \]

\noindent We define a notion of projection of a sequence of moves $s$
onto a
move $m$ inductively as follows:
\[ \begin{array}{lcll}
\varepsilon \, \restrict \, m & = & \varepsilon & \\
m[m'] \cdot s \, \restrict \, m & = & m' \cdot (s \restrict m) &  \\
m' \cdot s \, \restrict \, m & = & s \restrict m, & \forall m''. \, m' \neq
m[m''].
\end{array} \]

\noindent Dually, given an unambiguous set of moves $S$, and a
sequence of moves $s$ in which every move has the form $m[m']$ for some $m \in
S$ (necessarily unique since $S$ is unambiguous), we define a
projection $s \restrict S$ inductively as follows:
\[ \begin{array}{lcll}
\varepsilon \, \restrict \, S & = & \varepsilon & \\
m[m'] \cdot s \, \restrict \, S & = & m \cdot (s \restrict S) &  (m
\in S \; \wedge \; \rho (m) >0) \\
m[m'] \cdot s \, \restrict \, S & = & m[m'] \cdot (s \restrict S) &  (m
\in S^0 )
\end{array} \]

\subsection{Variable Games}
A \emph{variable game} is a structure
\[ A = (\Occ_A,  P_A , \gequiv_A ) \]
where:
\begin{itemize}
\item
$\Occ_A \subseteq \MM$ is an unambiguous  set of moves: the
\emph{occurrences} of $A$. We then define:
\begin{itemize}
\item $\lambda_A = \lambda \restrict \Occ_A$.
\item $\rho_A = \rho \restrict \Occ_A$.
\item $M_A = \{ m[m'] \mid m \in \Occ_A^0 \; \wedge \; m' \in \MM \}
  \;\; \cup \;\; \bigcup_{j > 0} \Occ_A^j$.
\end{itemize}

\item $P_A$ is a non-empty prefix-closed subset of $M_A^{\ast}$
  satisfying the following form of \emph{alternation condition}: the
  odd-numbered moves in a play are \emph{moves by $O$}, while the
  even-numbered moves are \emph{by $P$}. Here we regard the first, third,
  fifth, \ldots occurrences of a move $m$ in a sequence as being by
  $\lambda_A (m)$, while the second, fourth, sixth \ldots occurrences
  are by the other player.

\item $\gequiv_A$ is an equivalence relation on $P_A$ such that:
\[ \begin{array}{ll}
\mathbf{(e1)} & s \gequiv_A t \;\; \Longrightarrow \;\; s
\longleftrightarrow t \\
\mathbf{(e2)} & ss' \gequiv_A tt' \; \wedge \; | s | = | t | \;\;
\Longrightarrow \;\; s \gequiv_A t \\
\mathbf{(e3)} & s \gequiv_A t \; \wedge \; sa \in P_A \;\; \Longrightarrow
\;\; \exists b. \, sa \gequiv_A tb .
\end{array} \]
Here $s \longleftrightarrow t$ holds if
\[ s = \langle m_1 , \ldots , m_k \rangle, \qquad t = \langle m'_1 ,
\ldots , m'_k \rangle \]
and the correspondence $m_i \longleftrightarrow m'_i$ is bijective and
preserves $\lambda_A$ and $\rho_A$. We
write
\[ \pi : s \longleftrightarrow t \]
to give the name $\pi$ to the bijective correspondence $m_i
\longleftrightarrow m'_i$.
\end{itemize}

\noindent
A move $m \in \Occ_A^i$, $i > 0$, is an \emph{occurrence} of the type
variable $X_i$, while $m \in \Occ_A^0$ is a \emph{bound occurrence}.

The set of variable games is denoted by $\GG{\omega}$. The set of
those games $A$ for which the range of $\rho_A$ is included in $\{ 0,
\ldots , k \}$ is denoted by $\GG{k}$. Note that if $k \leq l$, then
\[ \GG{k} \subseteq \GG{l} \subseteq \GG{\omega} . \]
$\GG{0}$ is the set of \emph{closed games}.

\paragraph{Comparison with AJM games}
The above definition of game differs from that in \cite{AJM00} in
several respects.
\begin{enumerate}
\item The notion of bracketing condition, requiring a classification
  of moves as \emph{questions} or \emph{answers}, has been
  omitted. This is because we are dealing here with pure type
  theories, with no notion of ``ground data types''.
\item The alternation condition has been modified: we still have
  strict $OP$-alternation of moves, but now successive occurrences of
  moves within a sequence are regarded as themselves having
  alternating polarities. Since in the PCF games in \cite{AJM00} moves
  in fact only occur once in any play, they do fall within the present
  formulation. The reason for the revised formulation is that moves
  in variable games are to be seen as \emph{occurrences} of type variables,
  which can be expanded into plays at an instance. Another motivation
  comes from considering copy-cat strategies, in which (essentially)
  the same moves are
  played alternately by $O$ and $P$.

Technically, modifying the alternation condition in this way simplifies
the definition of substitution (see Section 3.4) and of the games
$X_i$ corresponding to type variables (see Section 3.5).

\item We have replaced the condition $\mathbf{(e1)}$ from \cite{AJM00}
  with a stronger condition, which is in fact satisfied by the games in
  \cite{AJM00}.
\end{enumerate}

\subsection{Constructions on games}
Since variable games are essentially just AJM games with some
additional structure on moves, the cartesian closed structure on AJM games
can be lifted straighforwardly to variable games.

\subsubsection*{Unit type}
The unit type $\mathbf{1}$ is the empty game.
\[ \mathbf{1} \;\; = \;\; ( \varnothing ,  \{ \varepsilon  \}
, \{ (\varepsilon , \varepsilon ) \} ) . \]

\subsubsection*{Product}
The product $A \& B$ is the disjoint union of games.
\[ \Occ_{A \& B} = \{ \al (m) \mid m \in \Occ_A \} \cup \{ \ar (m)
\mid m \in \Occ_B \} \]
\[ P_{A \& B} =  \{ \al^{\ast} (s) \mid s \in P_A \} \cup \{ \ar^{\ast} (t) \mid
t \in P_B \} \]
\[ \al^{\ast}(s) \gequiv_{A \& B} \al^{\ast}(t) \; \equiv \; s
\gequiv_A t \qquad \ar^{\ast}(s) \gequiv_{A \& B} \ar^{\ast}(t) \; \equiv \; s
\gequiv_B t . \]

\subsubsection*{Function Space}
The function space $A \Rightarrow B$ is defined as follows.

\[ \Occ_{A \Rightarrow B} \;\; = \;\; \{ \fl{i}(m) \mid i \in \omega
\; \wedge \; m \in \Occ_A \} \;\; \cup \;\; \{ \fr (m) \mid m \in
\Occ_B \} .  \]

\noindent $P_{A \Rightarrow B}$ is defined to be the set of all sequences in
$M_{A \Rightarrow B}^{\ast}$ satisfying the alternation condition, and
such that:
\begin{itemize}
\item $\forall i \in \omega . \, s \restrict \fl{i}(1) \in P_A$.
\item $s \restrict \fr (1) \in P_B$.
\end{itemize}
Let $S = \{ \fl{i}(1) \mid i \in \omega \} \cup \{ \fr (1) \}$. Note that
$S$ is unambiguous. Given a permutation $\alpha$ on $\omega$, we
define
\[ \breve{\alpha}(\fl{i}(1)) = \fl{\alpha (i)}(1), \qquad \breve{\alpha}(\fr
(1)) = \fr (1) . \]
The equivalence relation $s \gequiv_{A \Rightarrow B} t$ is defined by
the condition
\[ \exists \alpha \in S(\omega ) . \,
\breve{\alpha}^{\ast}(s \restrict S) = t \restrict S \; \wedge \;
s \restrict \fr (1) \gequiv_B t \restrict \fr (1)
\; \wedge \; \forall i \in \omega . \, s \restrict \fl{i}(1) \gequiv_A t
\restrict \fl{\alpha (i)}(1) ) .
\]
This is essentially identical to the definition in \cite{AJM00}. The
only difference is that we use the revised version of the alternation
condition in defining the positions, and that we define $A \Rightarrow
B$ directly, rather than via the linear connectives $\multimap$ and
$!$.

\subsection{Substitution}

Given $A \in \GG{k}$, and $B_1 , \ldots , B_k \in \GG{l}$, we define
$A[\vec{B}] \in \GG{l}$ as follows.

\[  \Occ_{A[\vec{B}]} \;\; = \;\; \Occ_A^0 \;\; \cup \;\;
\bigcup_{i=1}^k \, \{ m[m'] \mid m \in \Occ_A^i \; \wedge \; m'
\in \Occ_{B_i} \}. \]
\[
P_{A[\vec{B}]} = \{ s \in M_{A[\vec{B}]}^{\ast} \mid s \restrict A
  \in P_A
\; \wedge \;
\forall i: 1 \leq i \leq k. \, \forall m \in \Occ_A^i . \, s \restrict
m \in P_{B_i} \}
\]
\[
s \gequiv_{A[\vec{B}]} t \;\; \equiv \;\; s \restrict A \gequiv_A t
\restrict A
\; \wedge \; \pi : s \restrict A \longleftrightarrow t \restrict A \;\;
\Longrightarrow \;\;
\forall i: 1 \leq i \leq k. \, \forall m \in \Occ_A^i . \, s
\restrict m \gequiv_{B_i} t \restrict \pi (m) .
\]
Here by convenient abuse of notation we write $s \restrict A$ for $s
\restrict \Occ_A$.

\begin{proposition} $A[\vec{B}]$ is a well-defined game. In particular:
\begin{enumerate}
\item $\Occ_{A[\vec{B}]}$ is unambiguous.
\item $P_{A[\vec{B}]}$ satisfies the alternation condition.
\item $\gequiv_{A[\vec{B}]}$ satisfies \textbf{(e1)}--\textbf{(e3)}.
\end{enumerate}
\end{proposition}
\begin{proof}
\begin{enumerate}
\item This follows directly from Proposition~\ref{unamb}, since by assumption $\Occ_A$ and each $\Occ_{B_i}$ are unambiguous.

\item We begin by  formulating the alternation condition more precisely. We define the parity function
\[ \parity : \omega \longrightarrow \{ -1, +1 \}  \qquad \qquad \parity (k) = (-1)^k . \]
Also, for the purposes of this argument we shall interpret $P$ as $-1$ and $O$
 as $+1$. We can now define the alternation condition on a sequence $s$ as follows:
\[ \forall t \cdot m \sqsubseteq s. \, \parity (|t|) = \parity (\numoccs{m}{t})\lambda (m) . \]
We now consider a play $t \cdot m_1[m_2] \in P_{A[\vec{B}]}$.
Note firstly that if $\rho (m_1 ) = 0$, there is nothing more to prove, since in that case $t \cdot m_1[m_2] \restrict A = (t \restrict A) \cdot m_1[m_2]$ satisfies the alternation condition by assumption, and hence, since $|t| = |t \restrict A|$, so does $t \cdot m_1[m_2]$.

Otherwise, $\rho (m_1 )>0$. We shall use the following identities to verify the alternation condition for this play.
\[ \begin{array}{llcl}
(1) & |t| & = & |t \restrict A| \\
(2) & |t \restrict m_1 | & = & \numoccs{m_1}{t \restrict A} \\
(3) & \numoccs{m_1[m_2]}{t} & = & \numoccs{m_2}{t \restrict m_1} \\
(4) & \lambda(m_1[m_2]) & = & \lambda (m_1 ) \lambda (m_2 ) \\
(5) & \parity (|t \restrict A|) & = & \parity (\numoccs{m_1}{t \restrict A}) \lambda (m_1 ) \\
(6) & \parity (|t \restrict m_1 |) & = & \parity (\numoccs{m_2}{t \restrict m_1}) \lambda (m_2 ) .
\end{array} \]
Of these, (1)--(3) are easily verified; (4) follows from Proposition~\ref{sublr}; and (5) and (6) hold by assumption for plays in $A$ and each $B_i$ respectively. Now
\[ \begin{array}{lclr}
\parity (|t|) & = & \parity (|t \restrict A|) & (1) \\
& = & \parity (\numoccs{m_1}{t \restrict A}) \lambda (m_1 ) & (5) \\
& = & \parity (|t \restrict m_1 |) \lambda (m_1 ) & (2) \\
& = & \parity (\numoccs{m_2}{t \restrict m_1}) \lambda (m_1 ) \lambda (m_2 ) & (6) \\
& = & \parity (\numoccs{m_1[m_2]}{t}) \lambda (m_1[m_2]) & (3),(4)
\end{array} \]

\item We verify \textbf{(e3)}. Suppose that $s \gequiv_{A[\vec{B}]} t$ and $s\cdot m_1[m_2] \in P_{A[\vec{B}]}$.
This implies that $s \restrict A \gequiv_A t \restrict A$ and $(s
\restrict A)  \cdot m_1 \in P_A$. By \textbf{(e3)} for $A$, for
some $m'_1$, $(s \restrict A)\cdot m_1 \gequiv_A (t \restrict
A)\cdot m'_1$, and clearly if $\pi : (s \restrict A)\cdot m_1
\longleftrightarrow (t \restrict A)\cdot m'_1$, then $\pi (m_1 ) =
m'_1$. If $\rho (m_1) =0$, there is nothing more to prove.
Otherwise, if $m_1 \in \Occ_A^i$, $1 \leq i \leq k$, then $s
\restrict m_1 \gequiv_{B_i} t \restrict m'_1$, and $(s \restrict
m_1 )\cdot m_2 \in P_{B_i}$. By \textbf{(e3)} for $B_i$, for some
$m'_2$, $(s \restrict m_1 )\cdot m_2 \gequiv_{B_i} (t \restrict
m'_1 )\cdot m'_2$. Clearly $s \cdot m_1[m_2] \gequiv_{A[\vec{B}]}
t \cdot m'_1[m'_2]$, as required.
\end{enumerate}
\end{proof}

\subsubsection{Variants of substitution}
Firstly, note that the above definitions would still make sense if we
took $k = \omega$ and/or $l = \omega$, so that, for example, there is a
well-defined operation
\[ \GG{\omega} \times \GG{\omega}^{\omega} \longrightarrow \GG{\omega}
. \]
In practice, the finitary versions will be more useful for our purposes
here, as they correspond to the finitary syntax of System F.

More importantly, it is useful to define an operation of substitution
for one type variable only. We write this as
\[ A[B/X_i ] \]
where $B$ is being substituted for the $i$'th type variable $X_i$, $i
> 0$.

The definition is a simple variation on that of $A[\vec{B}]$ given
  above. Nevertheless, we give it explicitly, as we will make
  significant use of this version of substitution.

\[ \Occ_{A[B/X_{i}]} \;\; = \;\; \bigcup_{j \neq i} \Occ^j_{A} \; \cup \;  \{ m[m'] \mid m \in \Occ_A^i \; \wedge \; m' \in
\Occ_B \}.\]
\[
P_{A[B/X_{i}]} = \{ s \in M_{A[B/X_{i}]}^{\ast} \mid s \restrict A
  \in P_A
\; \wedge \;
\forall m \in \Occ_A^i . \, s \restrict
m \in P_{B} \}
\]
\[
s \gequiv_{A[B/X_{i}]} t \;\; \equiv
\;\; s \restrict A \gequiv_A t
\restrict A
\; \wedge \; \pi : s \restrict A \longleftrightarrow t \restrict A \;\;
\Longrightarrow \;\;
\forall m \in \Occ_A^i . \, s
\restrict m \gequiv_{B} t \restrict \pi (m) .
\]

\subsection{Properties of substitution}

\begin{proposition}
\label{sub00}
If $A \in \GG{k}$, $B_1 , \ldots , B_k \in \GG{l}$, and $C_1 , \ldots
, C_l \in \GG{m}$, then:
\[ A[B_1 [\vec{C}], \ldots , B_k [\vec{C}]] \; = \; (A[B_1 ,
\ldots , B_k ])[\vec{C}] . \]
\end{proposition}
\begin{proof}
We show firstly that
\[ \Occ_{A[B_1 [\vec{C}], \ldots , B_k [\vec{C}]]} \;\; = \;\; \Occ_{(A[B_1 ,
\ldots , B_k ])[\vec{C}]} . \]
Expanding the definitions, we can write the occurrence set of the LHS
of the equation as follows:
\[ \Occ_A^0 \; \cup \; \bigcup_i \Occ_A^i [\Occ_{B_i}^0] \; \cup \;
\bigcup_{i,j} \Occ_A^i [\Occ_{B_i}^j [\Occ_{C_j}]]  \] using the
notation $S[T] = \{ m_1[m_2] \mid m_1 \in S \; \wedge \; m_2 \in T
\}$.

Similarly, the occurrence set of the RHS can be expanded to
\[ \Occ_A^0 \; \cup \; (\bigcup_{i} \Occ_A^i [\Occ_{B_i}])^0 \; \cup
\; \bigcup_{j}((\bigcup_i \Occ_A^i [\Occ_{B_i}^j])[\Occ_{C_j}] .
\] Equating terms, the equality of these two sets follows from the
fact that $\rho (m[m']) = \rho(m')$, and hence $S[T]^i = S[T^i]$,
and that $m_1[m_2[m_3]] = (m_1[m_2])[m_3]$, and hence $S[T[U]] =
(S[T])[U]$..

Next we show that the conditions on plays on the two sides of the
equation are equivalent. Expanding the condition on plays on the LHS of the equation we see that $s \in P_{A[B_1 [\vec{C}], \ldots , B_k [\vec{C}]]}$ if:
\begin{enumerate}
\item $s \restrict A \in P_A$
\item $\forall i. \, \forall m \in \Occ_A^i . \, s \restrict m \restrict B_i \in P_{B_i}$
\item $\forall i. \, \forall j. \, \forall m \in \Occ_A^i . \,\forall m' \in \Occ_{B_i}^j . \, s \restrict m \restrict m' \in P_{C_j}$
\end{enumerate}
Similarly, expanding the condition on plays on the RHS yields:
\begin{enumerate}
\item $s \restrict A[\vec{B}] \restrict A \in P_A$
\item $\forall i. \, \forall m \in \Occ_A^i . \, s \restrict A[\vec{B}] \restrict m \in P_{B_i}$
\item $\forall j. \, \forall m \in \Occ_{A[\vec{B}]}^j . \, s \restrict m\in P_{C_j}$.
\end{enumerate}
Note firstly that for any $m \in \Occ_{A[\vec{B}]}^j$, for some $i$, $m = m_1[m_2]$ for $m_1 \in \Occ_A^i$, $m_2 \in \Occ_{B_i}^j$.
Now equating terms, we see that the equivalence of the two conditions is implied by the following equations:
\[ \begin{array}{ll}
1. &  s \restrict A[\vec{B}] \restrict A = s \restrict A \\
2. &  s \restrict A[\vec{B}] \restrict m = s \restrict m \restrict B_i \quad (m \in \Occ_A^i ) \\
3. & s \restrict m_1 \restrict m_2 = s \restrict m_1[m_2]
\end{array} \]
These equations are easily verified from the definitions of the
projection operations. Firstly, note that every move in these
games has the form (1) $m_1[m_2[m_3]]$, where for some $i$, $j$: $m_1 \in
\Occ_A^i$, $m_2 \in \Occ_{B_i}^j$, and $m_3 \in \Occ_{C_j}$; or the
form (2) $m_1[m_2]$, where $m_1 \in \Occ_A^0$; or (3) $m_1[m_2[m_3]]$, where $m_1 \in
\Occ_A^i$, $m_2 \in \Occ_{B_i}^0$. The LHS
of equation (1) projects  a move (1) firstly onto $m_1[m_2]$, then
onto $m_1$, whereas the RHS projects it directly onto $m_1$. Moves of
the form (2) are left unchanged in both cases; while moves of the
form (3) are projected onto $m_1$ in both cases.
In
equation (2), the effect of  the projection operations on both sides
of the equation is to
restrict the sequence to moves of the form $m[m_2[m_3]]$, and to
project each such move onto $m_2$. Finally, the effect of both sides
of equation (3) is to project $m_1[m_2[m_3]]$ onto $m_3$.

The argument for the coincidence of the equivalence relations is
similar.
\end{proof}

\noindent For each $i > 0$ we define the variable game $X_i$ as follows.
\[ \begin{array}{lcl}
\Occ_{X_i} & = & \{ i \} \\
P_{X_i} & = & M_{X_i}^{\ast} \\
s \gequiv_{X_i} t & \equiv & | s | = | t |
\end{array}
\]

\begin{proposition}
\label{sub0}
\begin{enumerate}
\item For all $B_1 , \ldots B_k \in \GG{\omega}$, $i \leq k$: $X_i [B_1 , \ldots B_k] \; = \; B_i$.
\item For all $A \in \GG{k}$: $A[X_1 , \ldots , X_k ] \; = \; A$.
\end{enumerate}
\end{proposition}

\begin{proposition}
\label{sub1}
The cartesian closed structure commutes with substitution:
\begin{enumerate}
\item $(A \Rightarrow B)[\vec{C}] \;\; = \;\; A[\vec{C}]
  \Rightarrow B[\vec{C}]$.
\item $(A \llwith B)[\vec{C}] \;\; = \;\; A[\vec{C}]
  \llwith B[\vec{C}]$.
\end{enumerate}
\end{proposition}
Combining Propositions~\ref{sub0} and \ref{sub1}, we obtain:

\begin{proposition}
The cartesian closed constructions can be obtained by substitution
from their generic forms:
\[ \begin{array}{lccc}
1. & A \Rightarrow B & = & (X_1 \Rightarrow X_2)[A, B] \\
2. & A \llwith B &  = &  (X_1 \llwith X_2) [A, B] .
\end{array} \]
\end{proposition}

\section{Constructing a Universe for Polymorphism}
\label{order}
\subsection{Two Orders on Games}
We will make use of two partial orders on games.

\begin{itemize}
\item The \emph{approximation order} $A \domapprox B$. This will be
  used in constructing games as solutions of recursive equations.

\item The \emph{inclusion order} $A \ginc B$. This will be used to
  define a notion of ``subgame'' within a suitable ``universal game''
  in our construction of a model of System F.
\end{itemize}

\subsubsection{The Approximation Order}
We define $A \domapprox B$ if:
\begin{itemize}
\item $\Occ_A \subseteq \Occ_B$
\item $P_A = P_B \cap M_A^{\ast}$
\item $s \gequiv_A t \;\; \Longleftrightarrow \;\; s \in P_A \; \wedge \;
  s \gequiv_B t$
\end{itemize}
Thus if we are given $B$ and $\Occ_A \subseteq \Occ_B$, then $A$ is
completely determined by the requirement that $A \domapprox B$. Note
that if $A \domapprox B$ and $\Occ_A = \Occ_B$, then $A = B$.

This order was studied in the context of AJM games in \cite{AM95}, and the
theory of recursively defined games was developed there and shown to
work very smoothly, in direct analogy with the treatment of recursion
on Scott information systems \cite{Win93}. All of this theory carries over
to the present setting essentially unchanged. The main facts which we
will need can be summarized as follows.
\begin{proposition}
\label{domapprox}
\begin{enumerate}
\item $(\GG{\omega}, {\domapprox})$ is a (large) cpo, with least upper
  bounds of directed sets being given by componentwise unions.
\item All the standard constructions on games, in particular
  product and function space, are monotonic and continuous with
  respect to the approximation order.
\item If a function $\GG{\omega} \longrightarrow \GG{\omega}$ is
  $\domapprox$-monotonic, and continuous on move-sets, then it is
  $\domapprox$-continuous.
\end{enumerate}
\end{proposition}

\noindent Thus if
\[ F : (\GG{\omega}, {\domapprox}) \longrightarrow (\GG{\omega},
{\domapprox}) \]
is continuous, we can solve the recursive equation
\[ X \;\; = \;\; F(X) \]
using the least fixed point theorem in the standard fashion to
construct a least solution in $\GG{\omega}$.

\subsubsection{The Inclusion Order}
We define $A \ginc B$ by:
\begin{itemize}
\item $\Occ_A \subseteq \Occ_B$
\item $P_A \subseteq P_B $
\item $s \gequiv_A t \;\; \Longleftrightarrow \;\; s \in P_A \; \wedge \;
  s \gequiv_B t$
\end{itemize}
Thus the only difference between the two orders is the condition
on plays. Note that
\[ A \domapprox B \;\; \Longrightarrow \;\; A \ginc B . \]
The inclusion order is useful in the following context. Suppose we fix
a ``big game'' $\UU$ to serve as a ``universe''. Define a
\emph{sub-game} of $\UU$ to be a game of the form
\[ A = (\Occ_{\UU}, P_A , \gequiv_{\UU} \cap P_A^2 ) , \]
where $P_A \subseteq P_{\UU}$, and
\[ s \in P_A \; \wedge \; s \gequiv_{\UU} t \;\; \Longrightarrow \;\; t \in P_A . \]
Thus sub-games of $\UU$ are completely determined by their sets of
positions. We write $\SubU$ for the set of sub-games of
$\UU$. Note that, for $A, B \in \SubU$:
\[ A \ginc B \;\; \Longleftrightarrow \;\; P_A \subseteq P_B . \]
\begin{proposition}
\label{SubU}
\begin{enumerate}
\item $\SubU$ is a complete lattice, with meets and joins given by
  intersections and unions respectively.
\item If $S \subseteq P_{\UU}$, then the least sub-game $A \in \SubU$
  such that $S \subseteq P_A$ is defined by
\[ P_A = \{ u \mid \exists s \in S. \, \exists t. \, t \sqsubseteq s
\; \wedge \; u \gequiv_{\UU} t \} . \]
\end{enumerate}
\end{proposition}

\noindent It is straightforward to verify that function space and product are
monotonic with respect to the inclusion order. This leads to the
following point, which will be important for our model construction.
\begin{proposition}
\label{Uclos}
Suppose that $\UU$ is such that
\[ \UU \Rightarrow \UU \domapprox \UU, \qquad \UU \llwith \UU \domapprox
\UU , \qquad \mathbf{1} \domapprox \UU . \]
Then $\SubU$ is closed under these constructions.
\end{proposition}
\begin{proof} Firstly,
\[ A, B \in \SubU \;\; \mbox{implies} \;\; A \Rightarrow B
  \ginc \UU \Rightarrow \UU , \]
by $\ginc$-monotonicity of $\Rightarrow$. But $\UU \Rightarrow \UU
\domapprox \UU$ by assumption, and since ${\domapprox} \subseteq
{\ginc}$, $A \Rightarrow B \ginc \UU$, \ie $A \Rightarrow B
\in \SubU$.
Similarly, $\SubU$ is closed under products.
\end{proof}

We also note the following for future reference.
\begin{proposition}
\label{submon}
Substitution $A[B_1 , \ldots , B_k ]$ is both $\ginc$-monotonic and $\domapprox$-monotonic in $A$ and each $B_i$, $1 \leq i \leq k$.
\end{proposition}
\begin{proof}
We show $\domapprox$-monotonicity for plays. Suppose $A \domapprox A'$ and $\vec{B} \domapprox \vec{B'}$. If $s \in M^{\ast}_{A[\vec{B}]}$, then $s \restrict A = s \restrict A'$, and for $m \in \Occ_A^j$, $1 \leq j \leq k$, $s \restrict m \in M_{B_j}^{\ast}$, and hence, since $B_j \domapprox B'_j$,
\[ s \restrict m \in P_{B_j} \;\; \Longleftrightarrow \;\; s \restrict m \in P_{B'_j} . \]
\end{proof}

\paragraph{Adjoints of substitution}
Let $A$ be a variable game, and $s \in P_{A[\UU/X_{i}]}$.  We can use the
substitution structure to compute the \emph{least} instance $B$ (with
respect to $\ginc$)  such
that $s \in P_{A[B/X_{i}]}$. We define
\[ A_{i}^{\ast}(s) = \{ t \mid \exists u. \, \exists  m \in \Occ_A^i . \; t \gequiv u
\; \wedge \; u \sqsubseteq s \restrict m \} \]

\begin{proposition}
\label{inv}
With notation as in the preceding paragraph, let $B = A_{i}^{\ast}(s)$.
\begin{enumerate}
\item $s \in P_{A[B/X_{i}]}$.
\item $s \in P_{A[C/X_{i}]} \;\; \Longrightarrow \;\; B \ginc C$.
\end{enumerate}
\end{proposition}
\begin{proof}
Fix $s \in P_{A[\UU/X_{i}]}$. For $C \in \SubU$,,
\[ s \in P_{A[C/X_{i}]} \;\; \Longleftrightarrow \;\; \{ s \restrict m
\mid m \in \Occ_A^i \} \subseteq P_C . \]
By Proposition~\ref{SubU}(2), $A_{i}^{\ast}(s)$ is the least $B \in
\SubU$ containing this set.
\end{proof}

\subsection{The Relative Polymorphic Product}
Given $A, B \in \GG{\omega}$ and $i > 0$, we define the relative
polymorphic product $\Pi_{i} (A, B)$ (the ``second-order quantification over $X_i$ in the
variable type $A$ relative to the universe $B$'') as follows.
\[ \Occ_{\Pi_{i}(A, B)} = \Occ_{A}[0/i] = \{ m[0/i] \mid m \in \Occ_{A} \} . \]
\[
P_{\Pi_{i}(A, B)} = \{ s \in P_{A[B/X_{i}]}  \mid  \forall t \cdot a
\sqsubseteq^{\mbox{{\scriptsize even}}} s.  \;
A_{i}^{\ast}(t \cdot a) = A_{i}^{\ast}(t) \}
\]
\[ s \gequiv_{\Pi_{i}(A, B)} t \;\; \Longleftrightarrow \;\; s
\gequiv_{A[B/X_{i}]} t . \]
To understand the definition of $P_{\Pi_{i}(A, B)}$, it is helpful to
consider the following alternative, inductive definition (\textit{cf.}
\cite{Abr96}):
\[ \begin{array}{llll}
P_{\Pi_{i}(A, B)} &=& & \{\epsilon\} \\
 & & \cup & \{ sa \mid s\in P_{\Pi_{i}(A, B)}^{\mbox{{\scriptsize even}}} \; \wedge\; \exists C \in \SubB . \,
sa\in P_{A[C]} \}\\
        & & \cup & \{ sab \mid sa\in  P_{\Pi_{i}(A, B)}^{\mbox{{\scriptsize
       odd}}} \;
\wedge \; \forall C \in \SubB .\, sa\in P_{A[C]}\ \Rightarrow \ sab\in P_{A[C]}
 \}
\end{array}
\]
The first clause in the definition of $P_{\Pi(F)}$
is the basis of the induction.
The second clause refers to positions in which it is Opponent's turn to move.
It says that Opponent may play in any way which is valid in {\em some}
instance. The final clause refers
to positions in which it is Player's turn to move.
It says that Player can only move in a fashion which is valid in {\em every}
possible instance. The equivalence of this definition to the one given
above follows easily from Proposition~\ref{inv}.

Intuitively, this definition says that initially, nothing is known about which instance we are playing in.
Opponent progressively reveals the ``game board'' ;
at each stage, Player is constrained to play within the instance
{\em thus far revealed} by Opponent.

The advantage of the definition we have given above is that it avoids
quantification over subgames of $B$ in favour of purely local
conditions on the plays.

\begin{proposition}
\label{sub2}
The relative polymorphic product commutes with substitution.
\begin{enumerate}
\item $\Pi_{i}(A, B)[C/X_{i}] = \Pi_{i}(A, B)$.
\item If $A \in \GG{k+1}$ and $C_1 , \ldots , C_k \in \GG{n}$,  then:
\[ \Pi_{k+1}(A, B)[\vec{C}] \;\; = \;\; \Pi_{n+1}(A[\vec{C}, X_{n+1}], B) . \]
\end{enumerate}
\end{proposition}
\begin{proof}
We prove (2). Firstly, we compare the occurrence sets.
Expanding the definitions on the LHS of the equation, we obtain
\[ \Occ_A^0 \; \cup \; \Occ_A^{k+1}[0] \; \cup \; \bigcup_{i=1}^k \{ m_1[m_2] \mid m_1 \in \Occ_A^i \; \wedge \; m_2 \in \Occ_{C_i} \} \]
Similarly, on the RHS we obtain
\[ \Occ_A^0 \; \cup \; \bigcup_{i=1}^k \{ m_1[m_2] \mid m_1 \in \Occ_A^i \; \wedge \; m_2 \in \Occ_{C_i} \} \; \cup \; \Occ_{A[\vec{C},X_{n+1}]}^{n+1}[0] \]
Since $\Occ_{A[\vec{C},X_{n+1}]}^{n+1}[0] = \Occ_A^{k+1}[n+1][0] = \Occ_A^{k+1}[0]$, we conclude that these two sets are equal.

We now show the equivalence of the conditions on plays. In similar fashion to the proof of associativity of substitution (Proposition~\ref{sub00}), this is a straightforward matter of expanding the definitions. The main point is to show the equivalence of the conditions restricting plays in the polymorphic products. This reduces to showing that
\[ A^{\ast}_{k+1}(s \restrict \Pi_{k+1}(A, B)) \; = \; A[\vec{C}, X_{n+1}]^{\ast}_{n+1}(s) , \]
which in turn reduces to showing that
\[ \{ s \restrict \Pi_{k+1}(A, B)  \restrict m \mid m \in \Occ_A^{k+1} \} \;\; = \;\; \{ s \restrict m[n+1] \mid m \in \Occ_A^{k+1} \} , \]
and finally to showing that for $m \in \Occ_A^{k+1}$,
\[ s \restrict  \Pi_{k+1}(A, B) \restrict m \; = \; s \restrict m[n+1] . \]
This holds because the projection $s \restrict \Pi_{k+1}(A, B)$ projects moves of the form  $m'[m'']$ with $m' \in \Occ_A^i$,  $1 \leq \rho_{A}(m') \leq k$, onto $m'$, and leaves the sub-sequence of elements of the form $m[m'']$ unchanged.
Finally, we note that projecting  with $m$ or $m[n+1]$ yields identical results.
\end{proof}

\begin{proposition}
\label{prodcont}
The relative polymorphic product $\Pi_{i}$ is
$\ginc$-monotonic and $\domapprox$-continuous as a function
\[ \GG{\omega} \times \GG{\omega} \; \longrightarrow \; \GG{\omega} . \]
\end{proposition}
\begin{proof}
For $\domapprox$-monotonicity, suppose $A \domapprox A'$ and $B \domapprox B'$. By Proposition~\ref{submon}, $A[B/X_{i}] \domapprox A'[B'/X_{i}]$. For $t \cdot a \in M_{A[B/X_{i}]}^{\ast}$, the further conditions on plays $C^{\ast}_i (t \cdot a) = C^{\ast}_i (t)$, for $C = A \; \mbox{or} \; A'$, depend only on the sets
\[ \{ u \restrict m \mid m \in \Occ_A^i \} , \quad u = t \; \mbox{or} \; t \cdot a \]
which depend only on $u$ and not on $C$.

For $\domapprox$-continuity, we use Proposition~\ref{domapprox}(3), by
which it suffices to show continuity on occurrence sets. The action of
$\Pi_i$ on occurrence sets is just that of substitution, which is defined
pointwise and hence preserves unions.
\end{proof}

\subsection{A Domain Equation for System F}

We define a variable game $\UU \in \GG{\omega}$ of System F types by the following
recursive equation:
\[ \UU \quad = \quad \&_{i>0} X_i \;\; \& \;\; \mathbf{1} \;\; \& \;\; (\UU \llwith \UU ) \;\; \& \;\; (\UU \Rightarrow \UU ) \;\; \& \;\; \&_{i>0}
\Pi_{i} (\UU , \UU ) . \]
Explicitly, $\UU$ is being defined as the least fixed point of a function $F : \GG{\omega} \longrightarrow \GG{\omega}$. This function is continuous by Propositions~\ref{domapprox} and~\ref{prodcont}.

We can then define second-order quantification by:
\[ \forall X_i . \, A \; \eqdef \; \Pi_{i} (A, \UU ) . \]
Although it is not literally the case that
\[ X_i \domapprox \UU, \qquad \UU \Rightarrow \UU \domapprox
\UU , \qquad \mbox{etc.} \]
for trivial reasons of how disjoint union is defined, with a
little adjustment of definitions we can arrange things so that we
indeed have
\[ \begin{array}{llcc}
\bullet & & & X_i \domapprox \UU   \\
\bullet & & & \mathbf{1} \domapprox \UU   \\
\bullet & A, B \domapprox \UU &
\Longrightarrow & A \llwith B \domapprox \UU \llwith \UU
\domapprox \UU \\
\bullet & A, B \domapprox \UU &
\Longrightarrow & A \Rightarrow B \domapprox \UU \Rightarrow \UU
\domapprox \UU \\
\bullet & A \domapprox \UU & \Longrightarrow & \forall X_i . \, A = \Pi_{i}(A,
\UU ) \domapprox \Pi_{i}(\UU ,
\UU ) \domapprox \UU .
\end{array} \]
Thus we get a direct inductive definition of the types of System F as
sub-games of $\UU$.

Moreover, if $A$ and $B$ are (the variable games corresponding to)
System F types, then a simple induction on the structure of $A$ using
Propositions~\ref{sub0}, \ref{sub1} and~\ref{sub2} shows
that
\[ A[B/X_{i}] \domapprox \UU , \]
and similarly for simultaneous substitution.

\section{Strategies}

Fix a variable game $A$. Let
\[ g : \Occ_A \pfnarrow \Occ_A \]
be a partial
function. We can extend $g$ to a partial function
\[ \hat{g} : M_{A[\vec{\UU}]} \pfnarrow M_{A[\vec{\UU}]} \]
by
\[
\hat{g}(m[m']) \; = \;  \left\{ \begin{array}{ll}
g(m)[m'], & g(m) \; \mbox{defined} \\
\mbox{undefined} & \mbox{otherwise}
\end{array} \right.
\]
Now we can define a set of plays $\sigma_g \subseteq M_{A[\vec{\UU}]}^{\ast}$
inductively as follows:
\[ \sigma_g \;\; = \;\; \{ \varepsilon \} \; \cup \; \{ sab \mid s \in \sigma_g
\; \wedge \; sa \in P_{A[\vec{\UU}]} \; \wedge \; \hat{g}(a) = b \}
  . \]
For all $\vec{B} \ginc \vec{\UU}$, we can define the restriction of
$\sigma_g$ to $\vec{B}$ by:
\[ \sigma_{\vec{B}} = \{ \varepsilon \} \; \cup \; \{ sab \in \sigma_g \mid sa
\in P_{A[\vec{B}]} \} . \]
(Note that
$\sigma_g = \sigma_{\vec{\UU}}$ in this notation.)
We say that $\sigma_g$ is a \emph{generic strategy} for $A$, and write
$\sigma_g : A$, if the following \emph{restriction condition} is satisfied:
\begin{itemize}
\item $\sigma_{\vec{B}} \subseteq P_{A[\vec{B}]}$ for all $\vec{B}
  \ginc \vec{\UU}$, so that the restrictions are well-defined.
\end{itemize}
Note that $\sigma = \sigma_g$ has the following properties.
\begin{itemize}
\item $\sigma$ is a non-empty set of even-length sequences, closed
  under even-length prefixes.
\item $\sigma$ is \emph{deterministic}, meaning that
\[ sab \in \sigma \; \wedge \; sac \in \sigma \;\; \Rightarrow \;\; b
= c. \]
\item $\sigma$ is \emph{history-free}, meaning that
\[ sab \in \sigma \; \wedge \; t \in \sigma \; \wedge \; ta \in P_{A[\vec{\UU}]}
\;\; \Rightarrow \;\; tab \in \sigma .\]
\item $\sigma$ is \emph{generic}:
\[ s \cdot m_1[m'_1] \cdot m_2[m'_2] \in \sigma \; \wedge \; t
\in \sigma \; \wedge \; t \cdot m_1[m''_1] \in P_{A[\vec{\UU}]}
\;\; \Rightarrow \;\; t \cdot m_1[m''_1] \cdot m_2[m''_1] \in
\sigma . \]
\end{itemize}
These conditions imply that
\[ s \cdot m_1[m'_1] \cdot m_2[m'_2] \in \sigma \; \Rightarrow
\; m'_1 = m'_2 ). \]
Moreover, for any set $\sigma \subseteq P_{A[\vec{\UU}]}$ satisfying
the above conditions, there is a
least partial function $g : \Occ_A \pfnarrow \Occ_A$ such that
$\sigma = \sigma_g$. This function can be defined explicitly by
\[ g(m_{1}) = m_2 \;\; \Longleftrightarrow \;\; \exists s. \, s \cdot m_1[a]
 \cdot m_2[a] \in \sigma . \]

The equivalence $\gequiv_A$ on plays can be lifted to a partial equivalence (\ie a symmetric and transitive relation) on strategies on $A$, which we also write as $\gequiv$. This is defined most conveniently in terms of a partial pre-order (transitive relation) $\preord$, which is defined as follows.
\[ \sigma \preord \tau \;\; \equiv \;\; sab \in \sigma \; \wedge \; t \in \tau \; \wedge \; sa \gequiv_A ta' \;\; \Longrightarrow \;\; \exists b' . \, ta'b' \in \tau \; \wedge \; sab \gequiv_A ta'b' . \]
We can then define
\[ \sigma \gequiv \tau \;\; \equiv \;\; \sigma \preord \tau \; \wedge \; \tau \preord \sigma . \]
A basic well-formedness condition on strategies $\sigma$ is that they
satisfy this relation, meaning $\sigma \gequiv \sigma$. Note that for
a generic strategy $\sigma = \sigma_{\vec{\UU}}$, using the
equivalence on plays in $A[\vec{\UU}]$:
\[ \sigma \gequiv \sigma \;\; \Longrightarrow \;\; \sigma_{\vec{B}} \gequiv
  \sigma_{\vec{B}} \quad \mbox{for all $\vec{B}
  \ginc \vec{\UU}$.} \]
A cartesian closed category of games is constructed by taking \emph{partial equivalence classes} of strategies, \ie strategies modulo $\gequiv$, as morphisms. See \cite{AJM00} for details.

\subsection{Copy-Cat Strategies}
One additional property of strategies will be important for
our purposes. A partial function $f : X \pfnarrow X$ is said to be a
\emph{partial involution} if it is symmetric, \ie if
\[ f(x) = y \;\; \Longleftrightarrow \;\; f(y) = x . \]
It is \emph{fixed-point free} if we never have $f(x) = x$.
Note that fixed-point free partial involutions on a set $X$ are in
bijective correspondence with pairwise disjoint families $\{ x_i ,
y_i \}_{i \in I}$ of two-element subsets of $X$ (\ie the set of
pairs $\{ x, y \}$ such that $f(x) = y$, and hence also $f(y) =
x$). Thus they can thought of as ``abstract systems of axiom links''.
See \cite{AL00,AL01} where a combinatory algebra of partial
involutions is introduced, and an extensive study is made of
realizability over this combinatory algebra.

For us, the important correspondence is with \emph{copy-cat strategies},
first  identified in \cite{AJ94a} as central to the game-semantical
analysis of proofs (and so-named there).
We say that $\sigma$ is a \emph{copy-cat strategy} if $\sigma =
\sigma_g$ where $g$ is a fixed-point free partial involution.

\begin{lemma}[The Copy-Cat Lemma]
Let $\sigma_g : A$ be a generic copy-cat strategy. If $g(m) = m'$, then for all
$s \in \sigma$:
\[ s \restrict m \;\; = \;\; s \restrict m' . \]
\end{lemma}
\begin{proof}
By induction on $|s|$. The base case is immediate.
Suppose that $s = t \cdot m_1[a] \cdot m_2[a]$ and that $g(m_3 ) =
m_4$. By the partial involution property of $g$,
\[ \{ m_1 , m_2 \} = \{ m_3 , m_4 \} \;\; \mbox{or} \;\;  \{ m_1 , m_2 \}
\cap \{ m_3 , m_4 \} = \varnothing . \]
In the first case,
\[ s \restrict m_1 = (t \restrict m_1 )\cdot a  = (t \restrict m_2 )
\cdot a  = s \restrict
m_2 , \]
where the middle equation follows from the induction hypothesis.

\noindent In the second case,
\[ s \restrict m_3 = t \restrict m_3 = t \restrict m_4 = s \restrict
m_4 , \]
where the middle equation again follows from the induction hypothesis.
\end{proof}

\subsection{Cartesian Closed Structure}
The required operations on morphisms to give the structure of a
cartesian closed category can be defined exactly as for AJM games
\cite{AJM00}. We give the basic definitions, referring to
\cite{AJM00} for motivation and technical details.

We write $\PInv{X}$ for the set of partial involutions on a set $X$.

\begin{proposition}
\begin{enumerate}
\item If $f \in \PInv{X}$ and $g \in \PInv{Y}$, then $f + g \in
  \PInv{X+Y}$.
\item If $f \in \PInv{Y}$, then $\ident_X \times f \in \PInv{X \times
    Y}$.
\item Partial involutions are closed under conjugation by
  isomorphisms:
\[ f \in \PInv{X} \; \wedge \; \alpha : X
  \stackrel{\cong}{\longrightarrow} Y \;\; \Longrightarrow \;\;
\alpha \circ f \circ \alpha^{-1} \in \PInv{Y} . \]
\end{enumerate}
\end{proposition}
Our basic examples of partial involutions will be
``twist maps'' (\ie symmetries) on
disjoint unions:
\[ \twist_X = [ \mathsf{in}_2 , \mathsf{in}_1 ] : X + X \longrightarrow X + X . \]
More generally, to get a partial involution on $\Occ_A$ we will
specify $\Occ'_A \subseteq \Occ_A$ and $\Occ_1, \ldots , \Occ_k$ such
that:
\[ \Occ'_A \;\; \cong \;\; (\Occ_1 + \Occ_1 ) +
\cdots + (\Occ_k + \Occ_k ) . \]
We then define a partial involution by conjugation by the indicated
isomorphism  of the evident disjoint union of $k$ twist maps. The partial involution is undefined on
$\Occ_A \setminus \Occ'_A$.

\paragraph{Identity}
For identity morphisms $\ident_A : A \Rightarrow A$,
\[ \Occ_{A \Rightarrow A} = \omega \times \Occ_A + \Occ_A . \]
Define $\Occ'_A = \{ 0 \} \times \Occ_A + \Occ_A \subseteq \Occ_{A
  \Rightarrow A}$.
Then
\[ \Occ'_A \cong \Occ_A + \Occ_A , \]
so we obtain the required partial involution as a twist map. This is
the basic example of a copy-cat strategy.

\paragraph{Projections}
Take for example $\pi_1 : A \llwith B \Rightarrow A$.
\[ \Occ_{A \llwith B \Rightarrow A} = \omega \times (\Occ_A + \Occ_B )
+ \Occ_A . \]
Define
\[ \Occ'_{A \llwith B \Rightarrow A} = \{ 0 \} \times (\Occ_A + \varnothing )
+ \Occ_A \subseteq \Occ_{A \llwith B \Rightarrow A}  . \]
Then $\Occ'_A \cong \Occ_A + \Occ_A$, and we obtain the required
partial involution by conjugating the twist map by the evident
isomorphism.

\paragraph{Pairing}
Suppose we are given partial involutions
\[ f \in \PInv{\Occ_{C \Rightarrow A}}, \qquad g \in \PInv{\Occ_{C
      \Rightarrow B}} . \]
\[ \Occ_{C \Rightarrow A \llwith B} = \omega \times \Occ_{C} +
\Occ_{A} + \Occ_{B} . \]
Using some bijection $\omega \cong \omega + \omega$,
\[ \begin{array}{rcl}
\Occ_{C \Rightarrow A \llwith B} & \cong & (\omega + \omega) \times \Occ_{C} +
\Occ_{A} + \Occ_{B} \\
& \cong & \omega \times \Occ_C + \omega \times \Occ_C  +
\Occ_{A} + \Occ_{B} \\
& \cong & (\omega \times \Occ_C  +
\Occ_{A}) + (\omega \times \Occ_C  +
\Occ_{B}) \\
& = & \Occ_{C \Rightarrow A} +  \Occ_{C \Rightarrow B} .
\end{array} \]
Then $f + g \in \PInv{\Occ_{C \Rightarrow A} +  \Occ_{C \Rightarrow
    B}}$, and conjugating by the indicated isomorphism yields the
required partial involution.

\paragraph{Application}
For application
\[ \mathsf{Ap}_{A, B} : (A \Rightarrow B) \llwith A \Rightarrow B , \]
\[ \begin{array}{rcl}
\Occ_A & = & \omega \times ((\omega \times \Occ_A + \Occ_B ) + \Occ_A )
+ \Occ_B \\
& \cong & \omega \times (\omega \times \Occ_A + \Occ_B ) + \omega
\times \Occ_A + \Occ_B \\
& \supseteq & \{ 0 \} \times (\omega \times \Occ_A + \Occ_B ) + \omega
\times \Occ_A + \Occ_B \\
& \cong & (\omega \times \Occ_A + \omega \times \Occ_A ) + (\Occ_B +
\Occ_B ) ,
\end{array} \]
yielding the required partial involution.

\paragraph{Currying}
Suppose that $f \in \PInv{\Occ_{A \llwith B \Rightarrow C}}$.
\[ \begin{array}{rcl}
\Occ_{A \llwith B \Rightarrow C} & = & \omega \times (\Occ_A + \Occ_B
) + \Occ_C \\
& \cong & \omega \times \Occ_A + (\omega \times \Occ_B + \Occ_C ) \\
& = & \Occ_{A \Rightarrow (B \Rightarrow C)} .
\end{array} \]
Conjugating $f$  by the indicated isomorphism yields the required
partial involution.

\paragraph{Composition}
Finally, we consider composition. We begin with some preliminaries on
partial involutions.
We write $\Rel{X}$ for the set of relations on a set $X$, \ie
$\Rel{X} = \pow{X \times X}$. Note that $\PInv{X} \subseteq
\Rel{X}$. We assume the usual regular algebra operations on relations:
composition $R \cdot S$, union $R \cup S$, and reflexive transitive
closure: $R^{\ast} = \bigcup_{k \in \omega} R^{k}$.

Any $R \in \Rel{X + Y}$ can be written as a disjoint union
\[ R = R_{XX} \cup R_{XY} \cup R_{YX} \cup R_{YY} , \]
where
\[ R_{ST} = \{ (a, b) \in R \mid a \in S \; \wedge \; b \in T \} . \]
Now given $R \in \Rel{X + Y}$, $S \in \Rel{Y + Z}$, we define $R
\invcomp S \in \Rel{X + Z}$ as follows:
\[ \begin{array}{rcl}
(R \invcomp S)_{XX} & = & R_{XX} \; \cup \;\; R_{XY} \cdot S_{YY} \cdot
(R_{YY} \cdot S_{YY})^{\ast} \cdot R_{YX} \\
(R \invcomp S)_{XZ} & = & R_{XY} \cdot (S_{YY} \cdot
R_{YY})^{\ast} \cdot S_{YZ} \\
(R \invcomp S)_{ZX} & = & S_{ZY} \cdot (R_{YY} \cdot
S_{YY})^{\ast} \cdot R_{YX} \\
(R \invcomp S)_{ZZ} & = & S_{ZZ} \; \cup \;\; S_{ZY} \cdot R_{YY} \cdot
(S_{YY} \cdot R_{YY})^{\ast} \cdot S_{YZ} .
\end{array} \]

\begin{proposition}
\label{invcomp}
\begin{enumerate}
\item $\invcomp$ is associative, with identity given by the twist map.
\item If $f \in \PInv{X + Y}$ and $g \in \PInv{Y+Z}$, then $f \invcomp
  g \in \PInv{X+Z}$.
\end{enumerate}
\end{proposition}
\begin{proof}
For (1), see \cite{AJ94a}. For (2), we write $\rconv{R}$ for relational
converse. Note that
\[ \rconv{(R \cdot S)} = \rconv{S} \cdot \rconv{R}, \quad \rconv{(R
  \cup S)} = \rconv{R} \cup \rconv{S}, \quad \rconv{(R^{\ast})} =
(\rconv{R})^{\ast}, \quad \rdconv{R} = R . \]
If $R \in \Rel{X+Y}$, then
\[ R = \rconv{R} \;\; \Longleftrightarrow \;\; \rconv{R_{XX}} = R_{XX}
\; \wedge \; \rconv{R_{XY}} = R_{YX} \; \wedge \; \rconv{R_{YX}} =
R_{XY} \; \wedge \; \rconv{R_{YY}} = R_{YY} . \]
Now if $R = \rconv{R}, S = \rconv{S}$:
\[ \begin{array}{rcl}
\rconv{(R \invcomp S)_{XX}} & = &  \rconv{(R_{XX} \; \cup \;\; R_{XY} \cdot S_{YY} \cdot
(R_{YY} \cdot S_{YY})^{\ast} \cdot R_{YX})} \\
& = & \rconv{R_{XX}} \; \cup \;\; \rconv{R_{YX}} \cdot (\rconv{S_{YY}} \cdot
\rconv{R_{YY}})^{\ast} \cdot
\rconv{S_{YY}} \cdot \rconv{R_{XY}} \\
& = & R_{XX} \; \cup \;\; R_{XY}  \cdot
(S_{YY} \cdot R_{YY})^{\ast} \cdot S_{YY} \cdot R_{YX} \\
& = & R_{XX} \; \cup \;\; R_{XY} \cdot S_{YY} \cdot
(R_{YY} \cdot S_{YY})^{\ast} \cdot R_{YX} \\
& = & (R \invcomp S)_{XX} ,
\end{array} \]
using the regular algebra identity $U \cdot (V \cdot U)^{\ast} = (U
\cdot V)^{\ast} \cdot U$.
The other cases are handled similarly.
\end{proof}

\noindent Now suppose we are given
\[ f \in \PInv{\Occ_{A \Rightarrow B}}, \qquad g \in \PInv{\Occ_{B
    \Rightarrow C}} . \]
\[ \Occ_{A \Rightarrow B} = \omega \times \Occ_A + \Occ_B \qquad
\Occ_{B \Rightarrow C} = \omega \times \Occ_B + \Occ_C . \]
Now $\ident_{\omega} \times f \in  \PInv{\omega \times (\omega \times \Occ_A
  + \Occ_B )}$, but using some bijection $\omega \cong \omega \times \omega$,
\[
\omega \times (\omega \times \Occ_A
  + \Occ_B )
\;\; \cong \;\; (\omega \times \omega ) \times \Occ_A + \omega \times \Occ_B
\\
\;\; \cong \;\; \omega  \times \Occ_A + \omega \times \Occ_B .
\]
Let $!f$ be the conjugation of $\ident_{\omega}{\times} f$ by the
indicated isomorphism. Then
\[ !f \invcomp g \in \PInv{\omega \times
  \Occ_A + \Occ_C} = \PInv{\Occ_{A \Rightarrow C}} \]
as required.

We show that composition is compatible with genericity at the level of
partial involutions. Recall that $\hat{g} = g \times
\ident_{M_{\UU}}$. Note that if $g \in \PInv{X}$, $\hat{g} \in \PInv{X
  \times M_{\UU}}$.
\begin{proposition}
If $f \in \PInv{X+Y}$ and $g \in \PInv{Y+Z}$, then
\[ (f \times \ident_U) \invcomp (g \times \ident_U) = (f \invcomp g)
\times \ident_U \in \PInv{(X + Z) \times U}. \]
\end{proposition}
\begin{proof}
This is immediate from the definition of $\invcomp$, since $-{\times}
\ident_U$ distributes over composition and union:
\[ (h \circ k) \times \ident_U = (h \times \ident_U ) \circ (k \times
\ident_U), \qquad (h \cup k) \times \ident_U = (h \times \ident_U ) \cup (k \times
\ident_U) . \]
\end{proof}

Next, we give a direct definition of composition on strategies as sets
of plays.
If $\sigma_g : A \Rightarrow B$ and $\sigma_h : B \Rightarrow C$, we
define
\[ \sigma_g ; \sigma_h = \{ s \restrict A,C \mid s \in (\omega \times
M_{A[\vec{\UU}]} + \omega \times
M_{B[\vec{\UU}]} + M_{C[\vec{\UU}]})^{\ast} \; \wedge \; s
\restrict A,B \in \sigma_{!g} \; \wedge \; s \restrict B,C \in
\sigma_h \} . \]

\begin{proposition}
$\sigma_g ; \sigma_h = \sigma_{!g \invcomp h}$.
\end{proposition}
\begin{proof}
See \cite{AJ94a}.
\end{proof}

\noindent Finally, we show the compatibility of composition with
restrictions to instances.
\begin{proposition}
$(\sigma ; \tau)_{\vec{D}} = \sigma_{\vec{D}} ; \tau_{\vec{D}}$.
\end{proposition}

\paragraph{Remark} The arbitrariness involved in the choice of
bijections $\omega \cong \omega + \omega$ and $\omega \cong \omega
\times \omega$ and the use of $0$ as  a particular element of $\omega$
in the above definitions is factored out by the partial equivalence
$\gequiv$, as explained in \cite{AJM00}.
Note that all the other ingredients used in constructing the above
isomorphisms are canonical, arising from the symmetric monoidal
structures of cartesian product and disjoint union on the category of sets, and the
distributivity of cartesian product over disjoint union. For the
general axiomatics of the situation, see \cite{AHS02}.

\section{The Model}

We shall use the hyper-doctrine formulation of model of System F, as
originally proposed by Seely \cite{See87} based on Lawvere's notion
of hyperdoctrines \cite{Law70}, and simplified by Pitts \cite{Pit88};
a good textbook presentation can be found in \cite{Cro93}.

We begin with a key definition:
\[ \GU{k} \;\; = \;\; \SubU \, \cap \, \GG{k} , \]
where $\UU$ is the universe of System F types constructed in Section~6.

\subsection{The Base Category}
We firstly define a base category $\Base$. The objects are natural
numbers. A morphism $n \longrightarrow m$ is an $m$-tuple
\[ \langle A_1 , \ldots , A_m \rangle , \qquad A_i \in \GU{n}, \; 1
\leq i \leq m . \] Composition of $\langle A_1 , \ldots , A_m
\rangle : n \longrightarrow m$ with $\langle B_1 , \ldots , B_n
\rangle : k \longrightarrow n$ is by substitution:
\[ \langle A_1 , \ldots , A_m \rangle \circ \vec{B} \;\; = \;\;
\langle A_1 [\vec{B}], \ldots , A_m [\vec{B}] \rangle : k
\longrightarrow m . \]
The identities are given by:
\[ \ident_{n} \;\; = \;\;  \langle X_1 , \ldots , X_n \rangle . \]
Note that variables act as projections:
\[ X_i : n \longrightarrow 1 \]
and we can define pairing by
\[ \langle \vec{A}, \vec{B} \rangle \;\; = \;\; \langle A_1 , \ldots ,
A_n , B_1 , \ldots , B_m \rangle : k \longrightarrow n+m \]
where
\[ \langle A_1 , \ldots , A_n \rangle : k \longrightarrow n, \qquad \langle B_1 , \ldots , B_m \rangle : k \longrightarrow
m . \]
Thus this category has finite products, and is generated by the object
$1$, in the sense that all objects are finite powers of $1$.

\subsection{The Indexed CCC}
Next, we define a functor
\[ \ICC : \Base^{\mbox{{\scriptsize op}}} \longrightarrow \CCC \]
where $\CCC$ is the category of cartesian closed categories with
\emph{specified} products and exponentials, and functors preserving
this specified structure.

The cartesian closed category $\ICC (k)$ has as objects $\GU{k}$. Note
that the objects of $\ICC (k)$ are the morphisms $\Base (k, 1)$; this
is part of the Seely-Pitts definition.

The cartesian closed structure at the object level is given by the
constructions on variable games which we have already defined: $A
\Rightarrow B$, $A \llwith B$, $\mathbf{1}$. Note that $\GU{k}$ is closed under these constructions by Proposition~\ref{Uclos}.

A morphism $A \longrightarrow B$ in $\ICC (k)$ is a generic copy-cat strategy
$\sigma : A \Rightarrow B$. Recall that this is actually defined at
the ``global instance'' $\UU$:
\[ \sigma = \sigma_{\UU} : (A \Rightarrow B)[\vec{\UU}] \;\; = \;\;
A[\vec{\UU}] \Rightarrow B[\vec{\UU}] . \]
More precisely, morphisms are partial equivalence classes of strategies modulo $\gequiv$.

The cartesian closed structure at the level of morphisms was
described in Section~5.2.

\subsubsection*{Reindexing}
It remains to describe the functorial action of morphisms in $\Base$. For
each $\vec{C} : n \rightarrow m$, we must define a cartesian closed
functor
\[ \vec{C}^{\ast} : \ICC (m) \longrightarrow \ICC (n) . \]
We define:
\[ \vec{C}^{\ast}(A) \; = \; A[\vec{C}] . \]
If $\sigma : A \Rightarrow B$,
\[ \vec{C}^{\ast}(\sigma ) = \sigma_{\vec{C}} : (A \Rightarrow
B)[\vec{C}] \;\; = \;\; A[\vec{C}] \Rightarrow B[\vec{C}] . \]
For functoriality, note that
\[ \vec{C}^{\ast}(\sigma ) \circ \vec{C}^{\ast}(\tau ) =
\sigma_{\vec{C}} \circ \tau_{\vec{C}} = (\sigma \circ \tau )_{\vec{C}} =
  \vec{C}^{\ast} (\sigma \circ \tau ) . \]
By Proposition~\ref{sub1}, $\vec{C}^{\ast}$ preserves the cartesian
closed structure.

\subsection{Quantifiers as Adjoints}
The second-order quantifiers are interpreted as right adjoints to
projections.
For each $n$, we have the projection morphism
\[ \langle X_1 , \ldots , X_n \rangle : n+1 \longrightarrow n \]
in $\Base$. This yields a functor
\[ \vec{X}^{\ast} : \ICC (n) \longrightarrow \ICC (n+1) . \]
We must specify a right adjoint
\[ \Pi_n : \ICC (n+1) \longrightarrow \ICC (n) \]
to this functor.
For $A \in \GU{n+1}$, we define
\[ \Pi_n (A) \; = \; \forall X_{n+1}. \, A . \]
To verify the universal property, for each $C \in \GU{n}$ we must
establish a bijection
\[ \Lambda : \ICC (n) (C, \forall X_{n+1}. \, A) \;\; \stackrel{\cong}{\longrightarrow} \;\; \ICC
(n+1)(\vec{X}^{\ast}(C), A) . \]
Concretely, note firstly that
\[ \vec{X}^{\ast}(C) \; = \; C[\vec{X}] \; = \; C . \]
Next, note that in both hom-sets the strategies are subsets of $P_{C[\vec{\UU}]
  \Rightarrow A[\vec{\UU},\, \UU/X_{n+1}]}$. In the case of generic strategies
  $\sigma$ into
  $A$, these are subject to the constraint of the \emph{restriction
    condition}: that is, for each instance $\vec{B},B$,
\[ \sigma_{\vec{B},B} \subseteq P_{C[\vec{B}] \Rightarrow
  A[\vec{B},B]} . \]
In the case of
  strategies $\sigma$ into $\forall X_{n+1}. \, A$, these are subject to the
  constraint that for each instance $\vec{B}$,
\[ \sigma_{\vec{B}} \subseteq P_{C[\vec{B}] \Rightarrow \forall
  X_{n+1}. \, A[\vec{B}, X_{n+1}]} . \]
Thus if we show that these conditions are equivalent, the required
correspondence between these hom-sets is simply the identity (which
also disposes of the naturality requirements)!

Suppose firstly that $\sigma$ satisfies the restriction
condition. Assuming that $sab \in \sigma$, we must show that
$A^{\ast}_{n+1}(sab) = A^{\ast}_{n+1}(sa)$. But if we let $B =
A^{\ast}_{n+1}(sa)$, then by Proposition~\ref{inv}(1),
\[ sa
\in P_{C[\vec{B}] \Rightarrow
  A[\vec{B},B]} , \]
and the restriction condition implies that
\[ sab
\in P_{C[\vec{B}] \Rightarrow
  A[\vec{B},B]} . \]
For the converse, suppose that $\sigma : C \Rightarrow \forall
X_{n+1}. \, A$. To show that $\sigma$ satisfies the restriction
condition, choose an instance $B$. Suppose that $sab \in \sigma$ and $sa \in
P_{C[\vec{B}] \Rightarrow A[\vec{B},B]}$. We must show that
$sab \in
P_{C[\vec{B}] \Rightarrow A[\vec{B},B]}$. Let $D =
    A^{\ast}_{n+1}(sa)$. Then by definition of $\forall X_{n+1}. \, A$,
    $sab \in A[\vec{B},D]$, and by Proposition~\ref{inv}(2), $D \ginc
    B$. Hence by Proposition~\ref{submon}, $sab \in
P_{C[\vec{B}] \Rightarrow A[\vec{B},B]}$ as required. $\;\; \Box$

\paragraph{Naturality (Beck-Chevalley)}
Finally, we must show that the family of right adjoints $\Pi_n$ form
an indexed (or fibred) adjunction. This amounts to the following: for
each $\alpha : m \longrightarrow n$ in $\Base$, we must show that
\[ \alpha^{\ast} \circ \Pi_n = \Pi_m \circ (\alpha \times
\ident_{1})^{\ast} . \]
Concretely, if $\alpha = \vec{C}$, we must show that for each $A \in
\GU{n+1}$,
\[ (\forall X_{n+1}. \, A)[\vec{C}] \;\; = \;\; \forall X_{m+1}. \,
A[\vec{C}, X_{m+1}] . \]
This is Proposition~\ref{sub2}.

\paragraph{Remark}
We are now in a position to understand the logical significance of the relative polymorphic product $\Pi_i (A, B)$. We could define
\[ \GB{k} \;\; =\;\; \SubB \; \cap \; \GG{k} , \]
and obtain an indexed category $\CCB (k)$ based on $\GB{k}$ instead of $\GU{k}$. We would still have an adjunction
\[  \GG{n} (C, \Pi_{n+1}( A, B)) \;\; \cong \;\; \CCB
(n+1)(\vec{X}^{\ast}(C), A) . \]
However, in general $B$ would not have sufficiently strong closure
properties to give rise to a model of System F. Obviously, $\SubB$
must be closed under the cartesian closed operations of product and
function space. More subtly, $\SubB$ must be closed under the
polymorphic product $\Pi_i ({-}, B)$. (This is, essentially, the
``small completeness'' issue \cite{Hyl88}, although our ambient category of
games does not have the requisite exactness properties to allow our
construction to be internalised in the style of realizability
models.\footnote{However, by the result of Pitts \cite{Pit88}, \emph{any}
hyperdoctrine model can be fully and faithfully embedded in an
(intuitionistic) set-theoretic model.}) This circularity, which directly reflects the impredicativity of System F, is resolved by the recursive definition of $\UU$.

\section{Homomorphisms}
We shall now view games as \emph{structures}, and introduce a natural
notion of homomorphism between games. These will serve as a useful
auxiliary tool in obtaining our results on genericity.

A homomorphism $h : A \longrightarrow B$ is a function
\[ h : P_A \longrightarrow P_B \]
which is
\begin{itemize}
\item \emph{length-preserving}: $|h(s)| = |s|$
\item \emph{prefix-preserving}: $s \sqsubseteq t \; \Rightarrow \;
  h(s) \sqsubseteq h(t)$
\item \emph{equivalence-preserving}: $s \gequiv t \; \Rightarrow \;
  h(s) \gequiv h(t)$.
\end{itemize}

There is an evident category $\GasS$ with  variable games as objects, and
homomorphisms as arrows.

\begin{lemma}[Play Reconstruction Lemma]
Let $A$, $B$ be variable games. If we are given $s \in P_A$,
and for each $m \in \Occ_A^i$, a play $t_m \in P_B$ with $|t_m | = \numoccs{m}{s}$,
then there is a unique $u \in P_{A[B/X_{i}]}$ such
  that:
\[ u \restrict A = s, \qquad u \restrict m = t_m  \;\; (m \in \Occ_A^i )
. \]
\end{lemma}
\begin{proof}
We can define $u$ explicitly by:
\[ \begin{array}{lcll}
u_j & = & s_j , & \rho_{A} (s_j ) \neq i \\
u_j & = & s_{j}[m], & \rho_{A} (s_j ) = i \; \wedge \; (t_{s_j})_k =
m,
\end{array} \]
where $j$ is the $k$'th position in $s$ at which $s_j$ occurs.
\end{proof}

This Lemma makes it easy to define a functorial action of variable
games on homomorphisms. Let  $A$ be a variable game, and $h : B
\longrightarrow C$ a homomorphism.
We define
\[ A(h) : A[B/X_{i}] \longrightarrow A[C/X_{i}] \]
by
$A(h)(s) = t$,
where
\[ t \restrict A = s \restrict A, \qquad t \restrict m = h ( s
\restrict m), \;\; (m \in \Occ_A^i ) . \]

\begin{lemma}[Functoriality Lemma]
$A(h)$ is a well-defined homomorphism, and moreover this action is
functorial:
\[ A(g \circ h) = A(g) \circ A(h), \qquad A(\ident_{B}) =
\ident_{A[B/X_{i}]} . \]
\end{lemma}

\noindent The second important property is that \emph{homomorphisms preserve plays of
generic strategies}.

\begin{lemma}[Homomorphism Lemma]
Let $A$ be a variable game,  $\sigma : A$ a generic strategy, and $h : C
\longrightarrow D$ a homomorphism. Then
\[ s \in \sigma_{A[C/X_{i}]} \;\; \Longrightarrow \;\; A(h)(s) \in
  \sigma_{A[D/X_{i}]} . \]
\end{lemma}
\begin{proof}
By induction on $|s|$. The base case is trivial. For the inductive
step, let
\[ u \equiv s \cdot m_1[a] \cdot m_2[a] \in \sigma_{A[C/X_{i}]}
. \]
By induction hypothesis, $A(h)(s) \in \sigma_{A[D/X_{i}]}$.
By the Copy-Cat Lemma, $u \restrict m_1 = u \restrict m_2$. Let $h(u
\restrict m_1 ) = v \cdot b$. Then
$A(h)(u) = A(h)(s) \cdot m_{1}[b] \cdot m_{2}[b]$, which is in
$\sigma_{A[D/X_{i}]}$ by genericity of $\sigma$.
\end{proof}

\section{Genericity}

Our aim in this section is to show that there are generic types in our
model, and indeed that, in a sense to be made precise, \emph{most
  types are generic}.

We fix a variable game $A \in \GG{1}$. Out aim is to find conditions
on variable games $B$ which imply that, for generic strategies $\sigma
, \tau : A$:
\[ \sigma_{B} \gequiv \tau_{B} \;\; \Longrightarrow \;\; \Lambda
(\sigma ) \gequiv \Lambda (\tau ) : \forall X. \, A . \] Since, as
explained in Section~5,
\[ \Lambda (\sigma ) \;\; = \;\; \sigma \;\; = \;\; \sigma_{\UU} , \]
this reduces to proving the implication
\[ \sigma_{B} \gequiv \tau_{B} \;\; \Longrightarrow \;\;
\sigma_{\UU} \gequiv \tau_{\UU} . \]
Our basic result is the following.

\begin{lemma}[Genericity Lemma]
\label{infgenericity}
If there is a homomorphism $h : \UU \longrightarrow B$, then $B$ is
generic.
\end{lemma}
\begin{proof}
We assume that $\sigma_{B} \gequiv \tau_{B}$, and show that
$\sigma_{\UU} \preord \tau_{\UU}$; a symmetric argument shows that $\tau_{\UU} \preord
\sigma_{\UU}$.

\noindent Suppose then that
\[ s \cdot m_{1}[a] \cdot m_{2}[a] \in \sigma , \quad t \in \tau ,
\quad s \cdot m_{1}[a] \gequiv t \cdot m'_{1}[a'] . \]
Let
\[ \begin{array}{lcl}
s' \cdot m_{1}[b] \cdot m_{2}[b] & = & A(h)(s \cdot m_{1}[a] \cdot
m_{2}[a]), \\
 t' \cdot
m'_{1}[b'] & = & A(h)(t \cdot m'_{1}[a']) .
\end{array} \]
Then since $A(h)$ is a homomorphism,
\[ s' \cdot m_{1}[b] \cdot m_{2}[b] \in P_{B} , \quad  t' \cdot
m'_{1}[b'] \in P_{B} , \quad s' \cdot
m_{1}[b]  \gequiv t' \cdot m'_{1}[b'] . \]
By the Homomorphism Lemma,
\[ s' \cdot m_{1}[b] \cdot m_{2}[b] \in \sigma , \quad t' \in \tau
. \]
Since by assumption $\sigma_{B} \gequiv \tau_{B}$, there exists
$m'_2$ such that:
\[ t' \cdot m'_{1}[b'] \cdot m'_{2}[b'] \in \tau \;\; \wedge \;\; s'
\cdot m_{1}[b] \cdot m_{2}[b] \gequiv t' \cdot m'_{1}[b'] \cdot
m'_{2}[b'] . \]
Since $\tau$ is generic, this implies that
\[ t \cdot m'_{1}[a'] \cdot m'_{2}[a'] \in \tau . \]
It remains to show that $s_1 \gequiv s_2$, where
\[ s_1 \equiv s \cdot m_{1}[a] \cdot m_{2}[a] , \quad s_2 \equiv t
\cdot m'_{1}[a'] \cdot m'_{2}[a'] . \]
Since by assumption
\[ s \cdot m_{1}[a] \gequiv t \cdot m'_{1}[a'] , \]
and $s'
\cdot m_{1}[b] \cdot m_{2}[b]  \gequiv t' \cdot m'_{1}[b'] \cdot
m'_{2}[b']$ implies that $s_1 \restrict A \gequiv s_2 \restrict A$,
it suffices to show that $s_1 \restrict m_2 \gequiv s_2 \restrict
m'_2$.
But by the Copy-Cat Lemma,
\[ s_1 \restrict m_2 = (s \restrict m_1 ) \cdot a , \quad s_2
\restrict m'_2 = (t \restrict m'_1 ) \cdot a' . \]
But
\[ s \cdot m_{1}[a] \gequiv t \cdot m'_{1}[a']  \;\; \Longrightarrow \;\; (s \restrict m_1 ) \cdot a \gequiv (t \restrict m'_1 ) \cdot a' , \]
and the proof is complete.
\end{proof}

\paragraph{Remark} The Genericity Lemma applies to \emph{any} variable
type $A$; in particular, it is \emph{not} required that $A$ be a
sub-game of $\UU$. Thus our analysis of genericity is quite robust,
and in particular is not limited to System F.

We define the \emph{infinite plays} over a game $A$ as follows:
$s \in P_A^{\infty}$ if every finite prefix of $s$ is in $P_A$.
We can use this notion to give a simple sufficient condition for the
hypothesis of the Genericity Lemma to hold.
\begin{lemma}
If $P_B^{\infty} \neq \varnothing$, then B is generic.
\end{lemma}
\begin{proof}
Suppose $s \in P_B^{\infty}$. Let $s_n \in P_B$ be the restriction of
$s$ to the first $n$ elements. We define $h : \UU \longrightarrow B$
by: $h(t) = s_{|t|}$. It is trivially verified that this is a
homomorphism. Genericity of $B$ then follows by the Genericity Lemma.
\end{proof}

We now apply these ideas to the denotations of System F types, the
objective being to show that ``most'' System F types denote generic
instances in the model.
Firstly, we define a notion of \emph{length} for games, which we then
transfer to types via their denotations as games.

\noindent We define
\[ |A| \;\; = \;\; \sup \{ |s| \mid s \in P_A \} . \]
Note that $|A| \leq \omega$.

We now show that any System F type whose denotation admits plays of
length greater than 2 is in fact generic!

\begin{lemma}[One, Two, Infinity  Lemma]
If $|T| > 2$, then $T$ is generic.
\end{lemma}
\begin{proof}
Consider the normal form of $T$, which can be written as
\[ \forall \vec{X}. \, T_1 \rightarrow \cdots \rightarrow T_k
\rightarrow X . \]
If $|T| \geq 3$, then there is a play of length three, in which the
first move must be made in the rightmost occurrence of $X$, the second
in a copy of some $T_i$ (by the definition of plays in the polymorphic product),
and the third must also be played in that same copy of $T_i$ (by the usual
switching conditions). But then the second and third moves can be
repeated arbitrarily often in different copies of $T_i$, giving rise
to an infinite play.
\end{proof}

We now give explicit syntactic conditions on System F types which
imply that they are generic.

\begin{proposition}
Let $T = \forall \vec{X}. \, T_1 \rightarrow \cdots \rightarrow T_k
\rightarrow X$.
\begin{enumerate}
\item If for some $i: 1 \leq i \leq k$, $T_i = \forall \vec{Y}. \, U_1 \rightarrow \cdots \rightarrow U_l
\rightarrow X$, then $T$ is generic.
\item If for some $i: 1 \leq i \leq k$, $T_i = \forall \vec{Y}. \, U_1 \rightarrow \cdots \rightarrow U_l
\rightarrow Y$, and for some $j: 1 \leq j \leq l$, $U_j = \forall \vec{Z}. \, V_1 \rightarrow \cdots \rightarrow V_m
\rightarrow W$, where $W$ is \textbf{either} some $Z_p \in \vec{Z}$,
\textbf{or} $Y$, \textbf{or} some $X_q
\in \vec{X}$, then $T$ is generic.
\end{enumerate}
\end{proposition}
\begin{proof}
It is easily seen that types of the shapes described in the statement
of the Proposition have
plays of length 3. Indeed in the first case $O$ plays in the rightmost
occurrence of $X$ in $T$, $P$ responds in the rightmost occurrence of $X$ in
the given $T_i$, and then $O$ can respond in that same occurrence of
$X$.
In the second case, $O$ plays in $X$, $P$ plays in $Y$, and then $O$
can play in $W$. We then apply the previous Lemma.
\end{proof}

\noindent We apply this to the simple and familiar case of ``ML types''.
\begin{corollary}
Let $T = \forall X. \, U$, where $U$ is built from the type variable
$X$ and $\rightarrow$. If $U$ is non-trivial (i.e. it is not just
$X$), then $T$ is generic.
\end{corollary}

\paragraph{Examples}
The following are all examples of generic types.
\begin{itemize}
\item $\forall X. \, X \rightarrow X$
\item $\forall X. \, (X \rightarrow X) \rightarrow X$
\item $\forall X. \, (\forall Y. Y \rightarrow Y \rightarrow Y)
  \rightarrow X$.
\end{itemize}

\paragraph{Non-examples}
The following illustrate the (rather pathological) types which do not
fall under the scope of the above results. Note that the
first two both have length 1; while the third has length 2.
\begin{itemize}
\item $\forall X. X$
\item $\forall X.  \forall Y. \, X \rightarrow Y$.
\item $\forall X. \, X  \, \rightarrow \, \forall X. \, X$
\end{itemize}

\paragraph{Remark} An interesting point illustrated by these examples is that our
conditions on types are orthogonal to the issue of whether the types
are inhabited in System F. Thus the type $\forall X. \, (X \rightarrow
X) \rightarrow X$ is not inhabited in System F, but is generic in the
games model, while the type $\forall X. \, X  \, \rightarrow \,
\forall X. \, X$ is inhabited in System F, but does not satisfy our
conditions for genericity.

\section{Full Completeness}

In this section, we prove full completeness for ML types.  The
full completeness proof exploits the decomposition of Intuitionist
implication into Linear connectives.  We give the basic
definitions, referring to \cite{AJM00} for motivation and
technical details.


\subsection{Linear Structure} The required operations on
morphisms to give the categorical structure required to model the
connectives of intuitionist multiplicative exponential linear
logic can be defined exactly as for AJM games \cite{AJM00}.

We fix an algebraic signature consisting of the following set of
\emph{unary} operations:
\[ \al , \; \ar , \; \{ \bangindex{i} \mid i \in \omega \} , \; \linearfl, \; \fr, \; \tensorfl, \; \tensorfr. \]
We take $\linearMM$ to be the algebra over this signature freely
generated by $\omega$. Explicitly, $\MM$ has the following
``concrete syntax'':
\[
m \;\; ::= \;\; i \; (i \in \omega ) \;\; \mid \;\; \al (m) \;\; \mid
\;\; \ar (m) \;\; \mid \;\; \bangindex{i} (m) \;\; (i \in \omega )
\;\; \mid \;\; \linearfl (m) \mid \;\; \fr (m) \mid
\;\;
  \tensorfl (m)\mid \;\; \tensorfr (m).
\]
The \emph{labelling map} $\lambda : \linearMM \longrightarrow \{
P, O \}$. The polarity algebra on the carrier $\{ P, O \}$
interprets $\al$, $\ar$, $\fr,\tensorfl,\tensorfr$ and each
$\bangindex{i}$ as the identity, and $\linearfl$ as the involution
$\bar{(\ )}$, where $\bar{P} = O$, $\bar{O} = P$. The map on the
generators is the constant map sending each $i$ to $O$.

\subsubsection*{Bang: $!$ }
$!A$ is defined as follows.

\[ \Occ_{!A } \;\; = \;\; \{ \bangindex{i}(m) \mid i \in \omega
\; \wedge \; m \in \Occ_A \} .  \]

\noindent $P_{!A }$ is defined to be the set of all sequences in
$M_{!A}^{\ast}$ satisfying the alternation condition, and such
that:
\begin{itemize}
\item $\forall i \in \omega . \, s \restrict \bangindex{i}(1) \in P_A$.
\end{itemize}
Let $S = \{ \bangindex{i}(1) \mid i \in \omega \} $. Given a
permutation $\alpha$ on $\omega$, we define
\[ \breve{\alpha}(\bangindex{i}(1)) = \bangindex{\alpha (i)}(1). \]
The equivalence relation $s \gequiv_{!A} t$ is
defined by the condition
\[ \exists \alpha \in S(\omega ) . \,
\breve{\alpha}^{\ast}(s \restrict S) = t \restrict S .
\]
This is essentially identical to the definition in \cite{AJM00}.
The only difference is that we use the revised version of the
alternation condition in defining the positions.

\subsubsection*{Linear function space: $A \linimpl B$}
The linear function space $A \multimap B$ is defined as follows.

\[ \Occ_{A \multimap B} \;\; = \;\; \{ \linearfl(m) \mid \; m \in
\Occ_A \} \;\; \cup \;\; \{ \fr (m) \mid m \in
\Occ_B \} .  \]

\noindent $P_{A \linimpl B}$ is defined to be the set of all
sequences in $M_{A \linimpl B}^{\ast}$ satisfying the alternation
condition, and such that:
\begin{itemize}
\item $s \restrict \linearfl (1) \in P_A$.
\item $s \restrict \fr (1) \in P_B$.
\end{itemize}
Let $S = \{\linearfl (1) , \fr (1) \}$. The equivalence relation $s \gequiv_{A \linimpl  B} t$ is defined
by the condition
\[ s \restrict S = t \restrict S \; \wedge \; s \restrict \fr (1) \gequiv_B t \restrict \fr (1) \; \wedge \;
s \restrict \linearfl (1) \gequiv_A t \restrict \linearfl(1) ) .
\]
\subsubsection*{Tensor: $A \tensor B$}
The tensor $A \tensor B$ is defined as follows.

\[ \Occ_{A \tensor B} \;\; = \;\; \{ \tensorfl(m) \mid \; m \in
\Occ_A \} \;\; \cup \;\; \{ \tensorfr (m) \mid m \in \Occ_B \} .
\]
 \noindent $P_{A \tensor B}$ is defined to be the
set of all sequences in $M_{A \tensor B}^{\ast}$ satisfying the
alternation condition, and such that:
\begin{itemize}
\item $s \restrict \tensorfl (1) \in P_A$.
\item $s \restrict \tensorfr (1) \in P_B$.
\end{itemize}
Let $S = \{\tensorfl (1) , \tensorfr (1) \}$. The equivalence relation $s \gequiv_{A \tensor B} t$ is defined by
the condition
\[ s \restrict S = t \restrict S \; \wedge \; s \restrict \tensorfr (1) \gequiv_B t \restrict \tensorfr (1) \; \wedge \;
s \restrict \tensorfl (1) \gequiv_A t \restrict \tensorfl(1) ) .
\]

\subsection{Domain equation}
Define the two orders $\ginc, \domapprox$ on games as before
(Section~\ref{order}).   We define a variable game $\LUU \in
\GG{\omega}$ of second order types by the following recursive
equation:
\[ \LUU =  \&_{i>0} X_i \;\; \& \;\; \mathbf{1} \;\; \& \;\; (\LUU
\llwith \LUU ) \;\; \& \;\; (\LUU \linimpl \LUU ) \;\; \& \;\; (\LUU
\tensor \LUU ) \;\; \& \;\; (!\LUU ) \;\; \& \;\; \&_{i>0}
\Pi_{i} (\LUU , \LUU ) . \]
Explicitly, $\LUU$ is being defined as
the least fixed point of a continuous function $F : \GG{\omega}
\longrightarrow \GG{\omega}$.

We first summarize the key facts required to relate $\UU$ and
$\LUU$.  Define $A \Rightarrow B \; = \; {!A \linimpl B}$.
\begin{proposition} \hfill
\begin{itemize}
\item  $\Rightarrow$, $\llwith$, Substitution and Relative Polymorphic
Product are all $\ginc$-monotone.
\item $(\SubLU, \ginc )$ is a complete lattice.
\end{itemize}
\end{proposition}
From this proposition, it is clear that $\UU$ is essentially a
subgame of $\LUU$, with the proviso that the universe of moves
underlying $\UU$ is different from the universe of moves in
$\LUU$.  More precisely, consider a \emph{renaming map} $R: \MM
\longrightarrow \linearMM$, that interprets $\al$, $\ar$, $\fr$ of
$\MM$ as the operations with the same name on $\linearMM$, and
$\fl{i}$ as $\linearfl \circ \bangindex{i}$. The map on the
generators is the ``identity'' map sending each $i \in \MM$ to $i
\in \linearMM$. Modulo this renaming map, $\UU$ is a subgame of
$\LUU$.

The genericity results for $\UU$ carry over to $\LUU$, in
particular the analog of lemma~\ref{infgenericity}.
\begin{lemma}\label{LUUgenericity} If
$P_B^{\infty} \neq \varnothing$, then B is generic.
\end{lemma}
In this light, since $\UU$ is essentially a subgame of $\LUU$, a
full completeness result for $\LUU$ implies full completeness for
$\UU$.

\subsection{Full completeness}

Consider an ML (universal closures of quantifier-free types) type
$T$, \ie $T = \forall \vec{X}. \, U$, where $U$ is quantifier-free.  
In the light of lemma~\ref{LUUgenericity}, it suffices to
prove the result when the type variables are instantiated with a
game $\iota$ such that $P_{\iota}^{\infty} \neq
\varnothing$. Explicitly, suppose that given a strategy $\sigma$ of
type $T$, we can find a term $M : T$ such that $\lsem M \rsem_{\iota}
\gequiv \sigma_{\iota}$. Then genericity implies that $\lsem M \rsem_{\LUU}
\gequiv \sigma_{\LUU}$, and hence that $\lsem M \rsem \gequiv \sigma$, as required.

We define $\iota$ as a well-opened subgame of $\LUU$, to enable us
to directly adopt the proofs from~\cite{AJM00}.  A game $B$ is
{\em well-opened}~\cite{AJM00} if the opening moves of $B$ can
\emph{only} appear as O-moves in opening positions. That is, for all 
$a\in M_B$, if
$a\in P_B$ then
\[ sa\in P_A \; \wedge \; | s |\; \mbox{even} \;\; \Longrightarrow \;\; 
s=\epsilon . \]

For notational convenience, we define $\iota$ as a subgame of
$\UU$.  By the earlier discussion, there is a variant of $\iota$
that is a subgame of $\LUU$.  Consider the System F type $(\forall
X) [(X \Rightarrow X) \Rightarrow X] $.  Let $n \in \omega$.
Consider the infinite position $s$ given by:
$$ \fr(n) \cdot \fl{1}(\fr(n)) \cdot \fl{1}(\fl{1}(n)) \cdot
\fl{2}(\fr(n)) \cdot \fl{2}(\fl{1}(n))
\cdot \fl{3}(\fr(n)) \ldots $$ Define $\iota$ as the minimum game
under the $\ginc$ order containing all the finite prefixes of $s$.
This is constructed as in Lemma~\ref{SubU}.  Explicitly, $\iota$ is
given  as  the set of positions that are equivalent to
finite prefixes of
$s$ in $\UU$.  An examination of the equivalence
in $\UU$ reveals that $\iota$ is well-opened.

Consider an ML type in which all type variables are instantiated
by $\iota$.  We now relate strategies in such types  to
$\beta\eta\Omega$-normal forms in the simply typed lambda calculus
built on a single base type $\iota$ with a constant $\Omega$ at
each type. $\Omega$ is interpreted in the model as the strategy
$\bot$ which only contains the empty sequence. For completeness,
we record the  $\beta\eta\Omega$-normal forms.
\begin{itemize}
\item For all types $T$, $\Omega : T$ and $x : T$  are $\beta\eta\Omega$ normal forms.
\item Let $M_i$ be $\beta\eta\Omega$-normal forms at types $T_i$, $1
  \leq i \leq k$.  Let $x$ be of type $T_1 \rightarrow \cdots \rightarrow
T_k \rightarrow \iota$.   Then $\lambda \vec{x}. \, x   M_1
\ldots M_k$ is a $\beta\eta\Omega$-normal form.
\end{itemize}

The statement of the decomposition theorem requires some further
notation from~\cite{AJM00}.   Consider
$$(A_1\llwith\dots\llwith A_k)\Rightarrow \iota$$ where
$$A_i=B_{i,1}\Rightarrow \dots B_{i,l_i}\Rightarrow \iota, \; \;\;
(1\leq i\leq l_i).$$ If for some $1\leq i\leq k$ and each $1\leq
j\leq l_i$ we have
$$\sigma_j:\tilde{A}\Rightarrow B_{i,j}$$ then we define
$${\bf C}_i(\sigma_1,\dots,\sigma_{l_i}):\tilde{A}\Rightarrow \iota$$ by
$${\bf
C}_i(\sigma_1,\dots,\sigma_{l_i})= {\tt Ap }\circ\lang \dots {\tt
Ap }\circ\lang \pi_i,\sigma_1\rang,\ldots,\sigma_{l_i}\rang .$$

With this notation, we are ready to state the decomposition lemma.
\begin{proposition}[Decomposition Lemma]\label{decomp}
Let $\sigma:(A_1\llwith\dots\llwith
A_p)\Rightarrow(A_{p+1}\Rightarrow\dots A_q\Rightarrow \iota)$ be
any strategy, where $$A_i=B_{i,1}\Rightarrow\dots
B_{i,l_i}\Rightarrow \iota,\;\; 1\leq i\leq q$$ We write
$\tilde{C}=A_1,\dots,A_p,\; \tilde{D}=A_{p+1},\dots,A_q$.
(Notation : if $\tau:\tilde{C},\tilde{D}\Rightarrow \iota$, then
$\Lambda_{\tilde{D}}(\tau):\tilde{C}\Rightarrow(A_{p+1}\Rightarrow\cdots
\Rightarrow A_q\Rightarrow \iota)$.)

Then exactly one of the following cases applies.
\begin{itemize}
\item[(i)] $\sigma=\Lambda_{\tilde{D}}(\bot_{\tilde{C},\tilde{D}}).$
\item[(ii)]
  $\sigma=\Lambda_{\tilde{D}}({\bf
C}_i(\sigma_1,\dots,\sigma_{l_i})),$ where $1\leq i\leq q$, and
\[\begin{array}{lllc}
\sigma_j:\tilde{C},\tilde{D} & \Rightarrow & B_{i,j}, & 1\leq
j\leq l_i
\end{array}\]
\end{itemize}
\end{proposition}
The proof follows standard arguments~\cite{AJM00,AL00}. In
particular, since $\iota$ is well-opened, the Bang Lemma
(Proposition 3.3.4 of~\cite{AJM00}) applies. The remainder of the
proof follows Proposition 3.4.5 of~\cite{AJM00}.

The Decomposition Lemma provides for one step of decomposition of
an arbitrary strategy into a form matching that of $\beta\eta\Omega$
normal forms in the $\lambdaomega$ calculus. However, infinite
strategies such as the $Y$ combinator  will not admit a well-founded
inductive decomposition process.

We conclude by describing a ``finiteness'' notion on strategies to
identify the strategies for which the decomposition terminates. We
define a notion of positive occurrences of $!$, following the
usual definition of positive occurrences of variables in a
formula. Consider a linear type built out of $!, \linimpl$ and
type variables.  We define positive and negative occurrences of
$!$ by structural induction.
\begin{itemize}
\item In $!A$, the positive occurrences of $!$ are the positive
occurrences in
$A$ and the outermost $!$.  The negative occurrences of $!$ in
$!A$ are the negative occurrences in $A$.
\item In $A \linimpl B$, the positive occurrences of $!$ are the
positive occurrences in $B$ and the negative occurrences in $A$.
The negative occurrences of $!$ are the negative occurrences in
$B$ and positive occurrences in $A$.
\end{itemize}
For any linear type $T$ built out of $!, \linimpl$ and $\iota$, consider $T'$ obtained by erasing the positive occurrences
of $!$ from $T$.  There is a canonical morphism $\delta_T : T \linimpl T'$
built by structural induction from dereliction maps (at the positive
occurrences of $!$) and identities (everywhere else). A strategy
$\sigma$ for an ML type $\forall X. F[X]$ is \emph{finite} if
there is a finite partial involution $f$ inducing
$\sigma_{\iota} ; \delta_{F[\iota ]}$.

The decomposition process is well-founded for \emph{finite}
strategies.
\begin{theorem}
For any ML type $\forall X. F[X]$, every finite strategy $\sigma$
is definable by a $\lambdaomega$ term in $\beta\eta\Omega$-normal form.
\end{theorem}

Stronger results can be proved, although we will not enter into details
here because of space restrictions. Firstly, if we extend the syntax of
$\lambda\Omega$ terms to allow \emph{infinite} terms (\ie we take the
ideal completion under the $\Omega$-match ordering), then we can
remove the finiteness hypothesis in the Theorem.
Secondly, if we refine the game model to introduce a notion of
winning infinite play, and use this to restrict to \emph{winning
strategies}, as in \cite{Abr96}, then we can obtain a full completeness
result for the ML types of System F itself, without any need to introduce $\Omega$
into the syntax.

\section{Related Work}
A game semantics for System F was developed by Dominic Hughes in his D.Phil. thesis
\cite{Hug99}. A common feature of his approach with
our's is that both give a direct interpretation of open types as certain games,
and of type substitution as an operation on games. However, his approach is in a sense rather closer to syntax; it
involves carrying type information in the moves, and the resulting model
is much more complex. For example, showing that strategies in the
model are closed under composition is a major undertaking. Moreover,
the main result in \cite{Hug99} is a full completeness theorem
essentially stating that the model is isomorphic to the term model of
System F (with $\beta\eta$-equivalence), modulo types being reduced to
their normal forms. As observed by Longo \cite{Lon95}, the term model
of System F \emph{does not satisfy Genericity}; in fact, it does not
satisfy Axiom (C). It seems that the presence of explicit type
information in the moves will preclude the model in \cite{Hug99} from
having  genericity properties comparable to those we have established for
our model.

The D.Phil thesis of Andrzej Murawski \cite{Mur01} takes a broadly similar approach
to modelling polymorphism to that of \cite{Hug99}, although the main
focus in \cite{Mur01} is on modelling Light Linear Logic.

\paragraph{Acknowledgements.}  Samson Abramsky was supported in part
by UK EPSRC GR/R88861. Radha Jagadeesan was supported in part by NSF
CCR-0244901.

\end{document}